\documentclass{article}
\usepackage{graphicx}
\usepackage{bbold}
\usepackage{slashed}
\usepackage{amsmath}
\usepackage{amssymb}
\usepackage{amsthm}
\usepackage{authblk}
\usepackage{biblatex}
\usepackage{geometry}
\usepackage{hyperref}
\usepackage{algorithm}
\usepackage{algpseudocode}
\usepackage{booktabs}
\usepackage{pdflscape}
\usepackage{appendix}
\usepackage{longtable}
\usepackage{siunitx}
\usepackage{multirow}%
\usepackage{amsfonts}%
\usepackage{amsthm}%
\usepackage{mathrsfs}%
\usepackage{xcolor}%
\usepackage{textcomp}%
\usepackage{manyfoot}%
\usepackage{algorithmicx}%
\usepackage{algpseudocode}%
\usepackage{listings}%
\newtheorem{proposition}{Proposition}

\newtheorem{lemma}{Lemma}

\begin{document}

\title{From Exponential to Polynomial: An Exact Filter for High-Dimensional MSM Models}
\author{D. A. Hameedi\\
Theoretical Particle Physics and Cosmology Group, Department of Physics\\
King's College London, London, WC2R 2LS, UK\\
\texttt{daniyal.hameedi@kcl.ac.uk}}
\maketitle

\begin{abstract}
    In this paper we propose a new formulation of the Bayesian Filter as used in the discrete-time Markov-Switching-Multifractal (MSM) model of volatility based on existing permutation symmetry within the likelihood structure. We show both analytically and empirically that such a formulation leads to a reduction in time complexity from $O(D^k)$ to $O(k^D)$ thereby significantly reducing the computational bottleneck associated with dimensionality. We compare the agreement between the naive and sector filters and find that while there are significant disagreements, the ground-truth recovery of the latter seems to improve on the former.
\end{abstract}

\tableofcontents

\section{Introduction}

Multifractal measures were first introduced by Mandelbrot in the study of turbulent flows in the 1970s \cite{Mandelbrot_1974}. Based on this work, Mandelbrot, Calvet, and Fisher would go on to formulate the Markov-Switching-Multifractal (MSM) model of Asset Returns in the late 90s \cite{mandelbrot1997multifractal}. The latter model provides us with a long memory methodology for computing volatility of financial assets. It has been shown that the MSM model provides much better results for both in and out of sample predictions, when compared to competing frameworks such as the GARCH derived models\cite{WANG20161}. In order to make such predictions and forecasts, the model requires calibration over historical asset price data, whereby the present market state is inferred through the process of bayesian filtering, and an associated logarithmic-likelihood is generated for known model components. One can then accordingly optimize the log-likelihood in order to then find the most suitable set of model parameters.

Despite its empirical success, the MSM framework is severely constrained by the exponential growth of its latent state space. In the binomial specification, exact filtering requires computational complexity $O(2^k)$ in the number of volatility components $k$, which generalizes to $O(D^k)$ in the multinomial case with $D$ possible discrete volatility states. This exponential scaling renders exact inference intractable for moderate to large values of $k > 13$. To address this issue, a range of approximate methods have been proposed, including the Generalized Method of Moments (GMM) \cite{lux2006markov} and sequential Monte Carlo (particle filtering) techniques\cite{LUX202269,calvet2012robustfiltering}. In this paper however, we will propose an exact reduction of this complexity, without need for approximation.

The filtering procedure is comprised of two steps, namely a predict and an update step \cite{caron2019hmm, Dymarski_2011} both of which have naive complexities of $O(D^k)$. In this paper we both, use a factorized representation of the full transition matrix \cite{ghahramani1997factorial, augustyniak2026mdsv}, and use symmetry arguments to develop an exact reformulation of the MSM filtering problem that reduces the leading computational complexity of the whole filtering procedure from $O(D^k)$ to $O(k^D)$, thereby replacing exponential dependence on the number of components with a polynomial scaling in $k$ (for fixed $D$). This reduction is achieved without approximation and therefore preserves the full information content of the original model. Our construction is based on the discrete timestep MSM formulation of Calvet and Fisher \cite{calvet_fisher_2004}. We hope too that the introduction of our reduced filter will allow for higher dimensional models creating more accurate and more holistic forecasts. 

The key idea underlying the reduction is the exploitation of existing permutation symmetry within the product-form likelihood of the model following a similar argument as \cite{kirkpatrick2013optimalstatespacereductionpedigree}. We can use this permutation invariance to construct equivalence class partitions by occupation number and configuration on the state space of the model. This technique is well understood in the mathematical and physical sciences where in the latter the occupation number representation of second quantization is a well understood and time-tested formalism in Quantum Many-Body-Systems \cite{10.1098/rspa.1927.0039, Fock1932KonfigurationsraumUZ, Roy_2024}. For a good review of that field, consider \cite{tafuri2026lecturenotesquantummanybody, tasaki2023introductionsecondquantizationformalism} and the canonical reference text \cite{10.1093/acprof:oso/9780199699322.001.0001}. In the former mathematical case the symmetric reduction of a state-space into a parametrization by its sufficient statistics has been well explored, particularly by Brockett, Mitter and Yau as part of their estimation algebra program \cite{brockett1980remarks, hazewinkel1983nonexistence, kang2024nonexistencefinitedimensionalestimationalgebras}, and also famously Kalman \cite{Kalman1960}. By constructing such equivalence classes, we can reformulate the filtering problem on a reduced state space of occupation number configurations, yielding an exact compression of the inference dynamics.

The remainder of the paper is structured as follows. Section 2 reviews the Markov-Switching Multifractal (MSM) model in the context of hidden Markov models and outlines the standard Bayesian filtering procedure. Section 3 derives the mathematical machinery required for sector filtering. Section 4 identifies the relevant symmetry structure and develops the associated representation-theoretic reduction of the state space, establishing the main theoretical complexity reduction result. Section 5 presents the resulting exact filtering algorithm and evaluates its performance against the standard formulation in tractable regimes with regards to time-complexity, exactness and comparability. Section 6 concludes with a discussion of our results, as well as potential avenues for improvement, investigation, and further applications. To implement the algorithm, and conduct the benchmarking we have developed a publicly available Python Library available here: \url{https://github.com/DAHameedi32/msm_filter}.

Finally, the Author would like to extend his gratitude to Professor Thomas Lux at the University of Kiel for productive discussions regarding this work.

\section{Discrete HMMs and The MSM Model}
Hidden Markov Models (HMMs) provide a natural framework for determining the underlying dynamics of stochastic processes, and are examples of finite learnable stochastic automata \cite{Dymarski_2011}.

We here provide a short review of the theory of Hidden Markov Models, following primarily the reviews laid out in \cite{Dymarski_2011} and \cite{caron2019hmm}, specifically looking at the discrete timestep case.
\subsection{Hidden Markov Models}
Consider the time series $X_{0:T} = (X_0, X_1, \dots, X_T)$, where each $X_i, i \in (0,T)$ is a random process taking a value within the state space $\mathcal{X}$. This time series obeys the Markov property, and is therefore called a Markov chain if for any $t \geq 0$ and $x_{0}, \dots, x_{t+1} \in \mathcal{X}$:
\begin{equation} \label{Markov_Property}
    \mathbb{P}(X_{t+1} = x_{t+1} |X_{t} = x_{t}, \dots, X_0 = x_0 ) = \mathbb{P}(X_{t+1} = x_{t+1} |X_{t} = x_{t}).
\end{equation}
Additionally, we can formulate a transition probability matrix $\textbf{T}$, which governs the time-evolution of the state $X_t$ to $X_{t+1}$. Indeed, we define it component-wise as such for $i, j \in \mathcal{X}$:
\begin{equation}
    T_{ij} = \mathbb{P}(X_{t+1} =j | X_t = i).
\end{equation}

We now associate with each Markov chain, the observation series $Y_{0:T} = (Y_1, \dots, Y_T)$ where $Y_i \in \mathcal{Y}$, the observation space. The observation series in general may be continuous or discrete, and we assume the observation series $(Y_1, \dots, Y_T)$ to be independent conditional on the state-sequence $(X_0, X_1, \dots, X_T)$. We then introduce the conditional probability on discrete elements:
\begin{equation}
    \mathbb{P}(Y_1 = y_1, \dots, Y_T =y_t | X_0 = x_0, \dots, X_T=x_T)= \prod_{t=1}^T \mathbb{P}(Y_t=y_t | X_t = x_t)
\end{equation}
Where we define the emission probability mass function:
\begin{equation}
    g_{x_t}(y_t) = \mathbb{P}(Y_t=y_t | X_t = x_t)
\end{equation}
also known as the likelihood function.
\subsection{Forward Filtering}
We are interested in the conditional probability mass function of the state $X_t$ given the observation data series $y_{1:t}$, which is given by $p(x_t|y_{1:t})$. By Bayes' Rule, this can be found as follows:
\begin{equation}
    p(x_t | y_{1:t}) = \frac{p(x_t, y_{1:t})}{\sum_{x_t' \in \mathcal{X}} p(x_t', y_{1:t})}
\end{equation}
Where $p(x_t, y_{1:t})$ is the joint distribution of the current state, and the historical data, which is then normalized over the joint distribution of all current states with the historical data.
We can then employ the following predict-update steps\cite{caron2019hmm}:
\begin{align}
    &\text{Predict:} \quad p(x_t | y_{1:t-1}) = \sum_{x_{t-1} \in \mathcal{X}}p(x_t|x_{t-1})p(x_{t-1}|y_{1:t-1})\\
    &\text{Update:} \quad p(x_t|y_{1:t}) =\frac{g_{x_t}(y_t) p(x_t|y_{1:t-1})}{\sum_{x_t' \in \mathcal{X}} g_{x'_t}(y_t) p(x'_t |y_{1:t-1})}
\end{align}
Indeed, the prediction step can be reformulated using the Transition Matrix $\textbf{T}$ as a linear map on the prior distribution $p(x_{t-1}|y_{1:t-1})$ using a vectorial representation of $\mathcal{X}$:
\begin{equation}
    p(x_t | y_{1:t-1}) = \textbf{T} \cdot p(x_{t-1}|y_{1:t-1})
\end{equation}

\subsection{The MSM Model}
The Markov-Switching-Multifractal (MSM) model of asset returns proposed by Calvet and Fisher in \cite{CALVET200127}, can be formulated as a discrete time series model, following \cite{calvet_fisher_2004}:
\begin{equation} \label{MSMModel}
    x_t = \mu_t + \sigma \left[\prod_{\ell=1}^{k} M_{\ell,t} \right]^{\frac{1}{2}} \varepsilon_t, \qquad\text{where }  \varepsilon_t \sim \mathcal{N}(0,1) \text{ i.i.d.} \,\, \forall\,\, t.
\end{equation}
Where $x_t$ is the total logarithmic return observed at time $t$, $\mu_t$ is the timestep drift of the model, which is interpreted as the dividend reinvestment, and $\sigma$ is the overall volatility scale. The product structure $\left[\prod_{i=1}^{k} M_{i,t} \right]$ is the product over the integer $k$ number of volatility components. The selection of $k$ is frequently viewed as a model selection problem, and an area of active research \cite{calvet_fisher_2004}. $M_{\ell, t}$ is drawn from a distribution $\mathcal{M}$ satisfying $\mathbb{E}(\mathcal{M}) =1$ for a general Markov chain,
with switching probability (the probability that the $\ell^{th}$ component changes from it's value at $t-1$ at time $t$) is given by $\gamma_\ell$ for a given $\ell \in 1, \dots,k$. In the binomial case where $M_{\ell,t} \in \{m_0, 2-m_0 \}$ for a convention where $m_0 > 1$ is the volatility high state upon which we condition such that: $\mathbb{P}(M_{\ell,t} = m_0) =\mathbb{P}(M_{\ell,t} = 2-m_0) =0.5$. If we generalize to the Multivariate case, we instead consider a spectrum of volatility component states: $\{m_0, m_1, \dots, m_D\}$, where each $m_i \in \mathbb{R}^+, \forall i \in 1, \dots,D$. In the literature it is common to consider the $k$-vector of multipliers at each timestep $M_t = (M_{1,t}, M_{2,t}, \dots, M_{k,t})$.

To update the model at each timestep, the $\ell^{th}$ component will either be left unchanged, or a new value will be drawn from $\mathcal{M}$ with a probability $\gamma_\ell$, called the transition probability. These transition probabilities follow a
geometric progression:
\begin{equation} \label{gamma_spectrum}
  \gamma_\ell = 1 - (1 - \gamma_1)^{b^{\ell -1}},
  \qquad \ell = 1, \ldots, k.
\end{equation}
Where, $\gamma_1 \in (0,1)$ and $b \in (1,\infty)$. We consider that each $M_{\ell,t}$ is interpreted to be a particular volatility shock over a given time period $\tau_\ell\propto 1/\gamma_\ell$ for each $\ell \in 1\dots k$, and that each of these components transition independently of each other component. In order to compute the full spectrum of transition probabilities, we require values of $\gamma_k$ and $b$. A common choice made in the literature is to take $\gamma_k \propto 1/T$ \cite{lux2006markov}, leaving us with the free parameter $b$ for determining the geometric series of transition probabilities.

We collect all the free parameters of the Binomial MSM into a state-vector of three values which will fully specify the model for a given value of $k$:
\begin{equation} \label{state-vector}
    \psi_k = (m_0, \sigma, b) \in \mathbb{R}^3_+.
\end{equation}
Determining this vector is an open problem in the literature, and we will broadly illustrate the approach in \cite{calvet_fisher_2004} of maximum likelihood estimation (MLE), which we detail in the following section.

\subsection{Model Selection via MLE and Bayesian Filtering}

Given a candidate value of $k$, the free parameters $\psi_k = (m_0, \sigma, b)$ are estimated by maximum likelihood, exploiting the fact that the forward filter of Section 1.2 yields the likelihood as a by-product of the recursion. Specifically, define the one-step predictive (normalizing) constant at time $t$:
\begin{equation}
    c_t(\psi_k) := \sum_{x_t' \in \mathcal{X}} g_{x_t'}(y_t)\, p(x_t' | y_{1:t-1}; \psi_k),
\end{equation}
which is precisely the denominator appearing in the Update step. Since $p(y_{1:T}) = \prod_{t=1}^T p(y_t | y_{1:t-1})$ by the chain rule of probability, and each factor is exactly $c_t(\psi_k)$, the log-likelihood of the observed return series under the model is
\begin{equation} \label{loglik}
    \ell(\psi_k) = \log L(\psi_k) = \sum_{t=1}^T \log c_t(\psi_k).
\end{equation}
The maximum likelihood estimator is then
\begin{equation}
    \hat{\psi}_k = \operatorname*{arg\,max}_{\psi_k \in \mathbb{R}^3_+} \ell(\psi_k).
\end{equation}
Because $\ell(\psi_k)$ has no closed form in $\psi_k$ it is only accessible pointwise, via a full forward pass through the data for each candidate value. This is a black-box numerical optimization problem. Following \cite{calvet_fisher_2004}, this is typically handled with a derivative-free search (e.g. Nelder–Mead, or grid search over $(m_0, b)$ followed by concentrating out $\sigma$ in closed form), since $\ell$ is known to be non-concave with a ridge-like structure that makes gradient-based methods unreliable. Model order $k$ itself is then chosen by comparing $\hat{\psi}_k$ across candidate $k$ via a likelihood-ratio, AIC/BIC, or out-of-sample forecasting criterion, since increasing $k$ always weakly increases in-sample likelihood.

\subsubsection{The Cost of the Forward Filter}

The bottleneck in this scheme is not the outer optimization but the inner filter itself. For the general (Multivariate) MSM in Eq.~\eqref{MSMModel}, each of the $k$ components $M_{\ell,t}$ takes one of $D$ values, so the joint state $X_t = (M_{1,t}, \dots, M_{k,t})$ lives in a state space of size
\begin{equation}
    |\mathcal{X}| = D^k.
\end{equation}
A single Predict step therefore requires forming and applying a $D^k \times D^k$ transition matrix, and naively costs $O(D^{2k})$ per timestep, or $O(T D^{2k})$ over the full sample — even before any outer loop over candidate $\psi_k$. This exponential blow-up in $k$ is precisely what limits practical MSM estimation to small $k$ (typically $k \leq 8$–$10$ in the literature, e.g. \cite{CALVET200127, lux2006markov}), despite the model's motivation being a cascade of arbitrarily many volatility components.

The transition structure does have exploitable sparsity: because each component switches independently with probability $\gamma_\ell$ given in Eq.~\eqref{gamma_spectrum}, the full transition kernel factorizes as a sum, over subsets $S \subseteq \{1,\dots,k\}$ of components that switch simultaneously, of Kronecker-structured termsm, thereby reducing the per-step cost from $O(D^{2k})$ to $O(k D^{k+1})$ by avoiding explicit matrix formation\cite{augustyniak2026mdsv,ghahramani1997factorial}. This remains exponential in $k$, however, and is the point of departure for the exact complexity reduction developed in the remainder of this paper: rather than filtering on $\mathcal{X} = \{m_0,\dots,m_D\}^k$ directly, we exploit the invariance of the binomial (or multinomial) MSM under the action of $S_k$ permuting the $k$ components, and show that the filtering recursion descends to the quotient space $\mathcal{X}/S_k$: the space of occupation-number configurations, of size $O(k^D)$ rather than $O(D^k)$, without loss of exactness.

\section{Exact Sector Filtering}
In this section we propose the exact sector filter following from two key properties of the general class of discrete timestep MSM models. The reduction relies on two key properties which hold true for the general class of discrete MSM models:
\begin{proposition}[Independence]
  The components $M_{1,t}, \ldots, M_{k,t}$ are mutually independent at each
  time $t$, both marginally and in their transition dynamics.
\end{proposition}

 \begin{proposition}[Product Emission] \label{Product_Emission}
  The conditional likelihood $\omega(x_t)_j$ depends on state $j$ only through
  $g(m^j) = \prod_\ell m^j_\ell$, the product of multipliers or some other likelihood structure symmetric under the permutation group $\mathcal{S}_k$.
\end{proposition}
Both of these properties hold to be true for the general class of the discrete MSM models, and are the key conceptual underpinnings of the reduction. The first property is true by construction of the MSM model \cite{calvet_fisher_2004}, and is a fundamental assertion of its dynamics. The second property is simple to prove: 
\begin{proof}
    We consider that for a given known hidden state $j = (m_{j_1}, m_{j_2}, \dots, m_{j_k}) \in \{m_1, m_2, \dots, m_D\}^k$ at time $t$, the only source of randomness in the equation \ref{MSMModel} is from $\varepsilon_t \sim \mathcal{N}(0,1)$, and is a linear transformation of a standard Gaussian, allowing us to use the standard result:
    \begin{equation}
        x_t | (\vec{M}_t = j) =\mu_t +  \sigma \left[\prod_{\ell=1}^{k} m_{j_\ell} \right]^{\frac{1}{2}} \varepsilon_t \sim \mathcal{N}\left(\mu_t, \sigma^2\left[\prod_{\ell=1}^{k}m_{j_\ell} \right]  \right).
    \end{equation}
    A proof of this result can be found in \cite{Soch2025}. 
    The emission pmf is then also given by the Gaussian:
    \begin{equation}
        g_{j}(x_t) =  \frac{1}{\sqrt{2\pi \, \sigma^2\left[\prod_{\ell=1}^{k}m_{j_\ell} \right]}}\exp{\left(-\frac{(x_t - \mu_t)^2}{2(\sigma^2\left[\prod_{\ell=1}^{k}m_{j_\ell} \right])}\right)}
    \end{equation}
    Where we can see that the emission pmf is dependent only on the state $j$ by the product of the volatility component states at time $t$: $\left[\prod_{\ell=1}^{k}m_{j_\ell} \right]$.

    We examine the product structure and note that for elements defined over the real number field, $ m_i \in \mathbb{R}, \quad  \forall i \in 1,\dots,D$, commutation under multiplication is trivial.
    \begin{equation}
        \therefore \forall g \in \mathcal{S}_k, \Rightarrow \left[\prod_{\ell=1}^{k}m_{g \cdot j_\ell} \right] = \left[\prod_{\ell=1}^{k}m_{j_\ell} \right] \square.
    \end{equation}
\end{proof}

Because the likelihood product has now been shown to be $\mathcal{S}_k$ symmetric, we can now define equivalence classes on arrangements of equal likelihood, which allows us to then lift to the space of these equivalence classes, which are indexed by the occupation vector $\vec{n} = (n_1, n_2, \dots, n_D)$. This lifting from the state-space to the occupation space occurs without any loss of underlying information, simply a reclassification of the original state-space into a more tractable and lower dimensional format for the filtering process.

\subsection{Factorized Prediction Step}

The factorization of the MSM transition matrix is well understood in the literature by \cite{ghahramani1997factorial} and also \cite{augustyniak2026mdsv}, but we present it here through the lense of statistical independence as part of further motivating the sector structure decomposition.

The full, unfactorized, transition matrix is of size $D^k \times D^k$, and acts on a $D^k$ state-vector. However, each component evolves independently of every other component, and rather than keeping track of the individual permutations, we can keep track of each of the $k$ states. In the Binomial MSM model, we had the following factorized form of the transition matrix:
\begin{equation}
    A = \bigoplus_{\ell = 1}^k A^{(\ell)}, \quad \text{Where } A^{(\ell)} =
    \begin{pmatrix}
        1 - \gamma_\ell/2 & \gamma_\ell/2 \\
        \gamma_\ell/2 & 1 - \gamma_\ell/2
    \end{pmatrix}.
\end{equation}
This factorized form was arrived at from the following form of the general case:
\begin{equation}
    T_{ij} = \bigoplus_{\ell = 1}^k \left[(1-\gamma_\ell)\delta_{m^{(\ell)}_i,m^{(\ell)}_j} + \gamma_\ell \mathbb{P}\left(\mathcal{M} = m^{(\ell)}_j\right)\right]
\end{equation}
Which we can easily check that the factorized version of the multivariate transition matrix will be as follows:
\begin{equation}
    \mathbf{T} = \bigotimes_{\ell}^k T_\ell, \quad T_\ell = \begin{pmatrix}
(1-\gamma_\ell)\pi_1 + \gamma_\ell & (1-\gamma_\ell)\pi_2 & \cdots & (1-\gamma_\ell)\pi_D \\
(1-\gamma_\ell)\pi_1 & (1-\gamma_\ell)\pi_2 + \gamma_\ell & \cdots & (1-\gamma_\ell)\pi_D \\
\vdots & \vdots & \ddots & \vdots \\
(1-\gamma_\ell)\pi_1 & (1-\gamma_\ell)\pi_2 & \cdots & (1-\gamma_\ell)\pi_D + \gamma_\ell

    \end{pmatrix} 
\end{equation}
Where each $T_n$ is a $D\times D$ matrix. This allows us to perform the prediction step in $O(kD^2)$.
\subsection{The Multinomial Case}
We consider that $\mathcal{M}$ may now possess a spectrum of values: $\{m_1, m_2, \dots, m_D\}$ from which a value is drawn at every time step. For an MSM model with $k$ volatility components this means a $D^k$ sized state space. However, following the permutation symmetry implied by property 2 above, we realize that we can reduce the size of the state space by partitioning it into equivalence classes labelled by the value of their conditional likelihood product: $\omega(x_t)_j$.

While in the Binomial MSM case, we considered the sector probability on a certain number of components being in the high state through the poisson binomial, we now generalize this to the poisson multinomial distribution (PMD). For a review of Poisson Multinomial Distributions see: \cite{Daskalakis_2015,liu2025maximumminimummultivariatepoisson,lin2022poissonmultinomialdistributionapplications}. This generalizes the sector structure from a simple scalar which quantifies the number of components in a high state to a $D-$vector quantity $\Vec{n}$ which tells us the occupation number of each state:
\begin{equation}
    \Vec{n} = (n_1, \dots, n_D)
\end{equation}
The Poisson Multinomial Distribution is then required to keep track of the probabilities for each occupation number:
\begin{equation}
    \mathbb{P}(n_1 = N_1, \dots, n_D = N_D)
\end{equation}
The PMD has the Generating Function:
\begin{equation}
    G(z_1, \dots, z_D) = \prod_{i=1}^k \left( \sum_{j=1}^D p_{ij}z_j \right)
\end{equation}
We consider that the system has a constraint: $\sum_i n_i = k$, which is the requirement that the number of components in the model does not change at any given point in time. Indeed, this constraint means that when $D-1$ component occupations are known, the $D^{th}$ component is completely determined. We can reflect this in the Generating function by setting $z_{D} = 1$, thereby making the $D^{th}$ occupation number redundant. This brings us to the following form of the generating function:
\begin{equation}
    G(z_1, \dots, z_{D-1}) = \prod_{i=1}^k \left( p_{iD} +\sum_{j=1}^{D-1} p_{ij}z_j \right) = \prod_i^{k}p_{iD} + \sum_{\vec{n}}q_{(\vec{n})}\prod_{n_i \in \vec{n}} z_i^{n_i}
\end{equation}
A more intuitive interpretation is that the imposition of the constraint constrains the system to the boundary of a volume within the state space. 

\subsection{The Leave-One-Out Distribution}
We can now formulate (in a similar fashion as with the binomial case) the generating function for the leave-one-out distribution of the PMD:
\begin{equation}
    G_{-\ell}(z_1, \dots, z_{D-1}) = \frac{G(z_1, \dots, z_{D-1})}{p_{\ell D} + \sum_{j}p_{\ell j} z_j} = \prod_{i=1\neq \ell} p_{iD} + \sum_{\vec{n}}Q_{(\vec{n}, \ell)}\prod_{n_i \in \vec{n}} z_i^{n_i}
\end{equation}
We can then use the power series formulation of the generating functions:
\begin{multline}
    G(z_1, \dots, z_{D-1}) = \left(p_{\ell D} + \sum_{j}^{D-1}p_{\ell j}z_j\right)G_{-\ell}(z_1, \dots, z_{D-1})\\
    \Rightarrow \prod_{i=1}^k p_{iD} + \sum_{\vec{n}}q_{(\vec{n})}\prod_{n_i \in \vec{n}}z^{n_i} = \left(p_{\ell D} + \sum_{j}^{D-1}p_{\ell j}z_j\right) \left(\prod_{i=1\neq \ell} p_{iD} + \sum_{\vec{n}}Q_{(\vec{n}, \ell)}\prod_{n_i \in \vec{n}} z_i^{n_i} \right)\\
    \Rightarrow \sum_{\vec{n}}q_{(\vec{n})}\prod_{n_i \in \vec{n}} z_{i}^{n_i} = \sum_j^{D-1} \left(\prod_{i=1\neq \ell}p_{iD}\right)p_{\ell j} z_j + \sum_{(\vec{n} \neq \hat{e}_j)}p_{\ell D} Q_{(\vec{n}, \ell)}\prod_{n_i \in \vec{n}} z_i^{n_i} + \sum_{j,\vec{n}\neq \hat{e}_j}p_{\ell j}Q_{(\vec{n}, \ell)}\prod_{n_i \in \vec{n}}z_i^{n_i}\\
    \Rightarrow \tilde{q}_{\vec{n}} \prod_{n_i \in \vec{n}}z_i^{n_i} = \sum_{j=1}^{D-1}\left(\prod_{i=1\neq\ell}^k\tilde{p}_{iD}\right)\tilde{p}_{\ell j} \delta_{\vec{n}, \hat{e}_j}\prod_{n_i \in \vec{n}}z_i^{n_i} +\tilde{p}_{\ell D} Q_{\vec{n}, \ell} \prod_{n_i \in \vec{n}}z_i^{n_i}+ \sum_{j}^{D-1}\tilde{p}_{\ell j}Q_{(\vec{n} - \hat{e}_j, \ell)}\left( \prod_{n_i \in \vec{n}}z_i^{n_i}\right)
\end{multline}
Where the linear term corresponds to the base case of the recursion, which occurs in the case where all components are in the $D^{th}$ state (at the boundary). We then make some rearrangements to factorize the linear term in $z_j$. We start by considering a particular arrangement $\vec{n} \neq \hat{e}_j$:
\begin{multline}
    \tilde{q}_{\vec{n}} \prod_{n_i \in \vec{n}}z_i^{n_i} = \tilde{p}_{\ell D} Q_{\vec{n}, \ell} \prod_{n_i \in \vec{n}}z_i^{n_i}+ \sum_{j}^{D-1}\tilde{p}_{\ell j}Q_{(\vec{n} - \hat{e}_j, \ell)}\left( \prod_{n_i \in \vec{n}}z_i^{n_i}\right)\\
    \Rightarrow Q_{\vec{n}, \ell} = \frac{1}{\tilde{p}_{\ell D}}\left[ \tilde{q}_{(\vec{n})} - \sum_j^{D-1}\left( \tilde{p}_{\ell j}Q_{(\vec{n}-\hat{e}_j, \ell)} \right) \right]
\end{multline}
Additionally, we can find the base case of the recursion by setting $\vec{n} = \hat{e}_j$, which  will simply yield the linear term, for consistency. This base case corresponds to the arrangement where all components are in the $D^{th}$ state except the $\ell^{th}$ one, which is in the $j^{th}$ state:
\begin{multline}
    Q_{(\hat{e}_j, \ell)} = \frac{1}{\tilde{p}_{\ell D}} \left[ \tilde{q}_{(\hat{e}_j)} - \sum_{j'}^{D-1}\left( \tilde{p}_{\ell j'}Q_{(\vec{0}, \ell)} \right) \right]\\
    \text{To reproduce the linear term: } Q_{(\vec{0}, \ell)} = \prod_{i\neq \ell} p_{iD}, \quad \Rightarrow Q_{(\hat{e}_j, \ell)} = \frac{1}{\tilde{p}_{\ell D}} \left[ \tilde{q}_{(\hat{e}_j)} -  \tilde{p}_{\ell j}\prod_{i \neq \ell}p_{iD} \right] \square.
\end{multline}
Where the sum collapses to a single term where $Q_{(\hat{e}_{j'} - \hat{e}_j, \ell)} = 0$ unless $j' = j$.
\subsection{Marginal Recovery Formula}
We consider the following definition of the Marginal Probability:
\begin{equation}
    p_{\ell j, t} = \mathbb{P}(M_{\ell,t} = m_j | x_0, x_1, \dots, x_t) = \sum_{\vec{n}} \mathbb{P}(M_{\ell,t} = m_j | \vec{n})q_{\vec{n}}
\end{equation}
Where $q_{\vec{n}}$ is the actual sector posterior distribution. We can define the predicted marginal with respect to the actual marginal as follows:
\begin{equation}
    q_{(\vec{n})} = \frac{\omega_{(\vec{n})}(x_t) \tilde{q}_{(\vec{n})}}{\sum_{\vec{n}}\omega_{\vec{n}}(x_t)\tilde{q}_{(\vec{n})}}
\end{equation}
We then decompose the probability $\mathbb{P}(M_{\ell, t} = m_j |\vec{n})$ into the following function:
\begin{equation}
    \mathbb{P}(M_{\ell, t} = m_j | \vec{n}) = \frac{\tilde{p}_{\ell j, t}Q_{(\vec{n}-\hat{e}_j, \ell)}}{\tilde{q}_{(\vec{n})}}
\end{equation}
Which we then plug back into the marginal:
\begin{equation}
    p_{\ell j, t} =\sum_{\vec{n}}\frac{\tilde{p}_{\ell j, t}Q_{(\vec{n}-\hat{e}_j, \ell)}}{\tilde{q}_{(\vec{n})}}\frac{\omega_{(\vec{n})}(x_t) \tilde{q}_{(\vec{n})}}{\sum_{\vec{n}}\omega_{\vec{n}}(x_t)\tilde{q}_{(\vec{n})}}
\end{equation}
Which can then be simplified as follows:
\begin{equation}
    p_{\ell j, t} = \tilde{p}_{\ell j, t}\frac{\sum_{\vec{n}}Q_{(\vec{n}-\hat{e}_j, \ell)}\omega_{(\vec{n})}(x_t)}{\sum_{\vec{n}}\omega_{\vec{n}}(x_t)\tilde{q}_{(\vec{n})}}
\end{equation}
We can then use the predicted sector marginal expression to make the denominator more explicit, and allow it to be found in terms of the recursion:
\begin{equation}
    p_{\ell j, t} = \tilde{p}_{\ell j, t}\frac{\sum_{\vec{n}}Q_{(\vec{n}-\hat{e}_j, \ell)}\omega_{(\vec{n})}(x_t)}{\sum_{\vec{n}}\omega_{\vec{n}}(x_t) \left[ \tilde{p}_{\ell D} Q_{(\vec{n}, \ell)} + \sum_{j}^{D-1}\tilde{p}_{\ell j}Q_{(\vec{n} - \hat{e}_j, \ell)} \right] }
\end{equation}

This marginal recovery formula, in conjunction with the efficient recursions described above allows the Bayesian inference to be computed in polynomial $O(k^D)$ rather than $O(D^k)$ time.

\section{Algorithmic Complexity}

To derive the Algorithmic Complexity we consider that we know a priori that the prediction step will be $O(kD^2)$. For the update step, we consider that one can use representation theoretic methods to present the underlying complexity reduction as a reduction in the dimensionality of the space upon-which the inference update step is carried out, and present this in the proceeding section.
We consider the original state-space $\Sigma = \{m_1, m_2, \dots, m_D \}^k \Rightarrow |\Sigma| = D^k$, which has the conditional likelihood defined on it: $\omega(x_t) = \prod_j^k m_j$. 
We further consider that under index permutation of the elements in the likelihood structure, the likelihood is invariant under the action of the permutation group $\mathcal{S}_k$:
 \begin{equation*}
        \therefore \forall g \in \mathcal{S}_k, \Rightarrow \left[\prod_{\ell=1}^{k}m_{g \cdot j_\ell} \right] = \left[\prod_{\ell=1}^{k}m_{j_\ell} \right].
\end{equation*}
We then group into equivalence classes of equivalent likelihood. This invariance is described by the symmetry group $\mathcal{S}_k$, the permutation group of $k$ elements. 
We can recall that when defining these equivalence classes, we are effectively quotienting the state-space by this finite group, forming the orbit space $\Sigma /\mathcal{S}_k$. 
This quotient space of equivalence classes is the space upon-which the subsequent inference is run, and we posit that the complexity of that inference process will be some function of the dimensionality of the reduced space.

We determine the dimensionality of the orbit space by Burnside's Lemma\cite{serre1977linear}:
\begin{equation} \label{Burnside}
    |\Sigma/\mathcal{S}_k|= \frac{1}{|\mathcal{S}_k|} \sum_{g \in \mathcal{S}_k} |Fix( \Sigma, g)|
\end{equation}
Where $g$ is an element of the group $\mathcal{S}_k$, and we know that $|\mathcal{S}_k|=k!$, as there are $k$ ways to pick the first element, $k-1$ ways to pick the second element, etc.
The function $|Fix(\Sigma,g)|$ denotes the number of elements in $\Sigma$ whose orbits are fixed under the action of $\mathcal{S}_k$, which in this case corresponds to having equivalent likelihoods.
We can use the following identity to make the change from summing over group elements to summing over states $s \in \Sigma$:
\begin{equation}
    \sum_{g \in \mathcal{S}_k} |Fix(\Sigma, g)| = \sum_{s \in \Sigma} |Stab(s)|
\end{equation}
Where the stabilizer $Stab(s)$ is the set of permutations which maps a state $s$ labeled by a given occupation vector $\vec{n}$ to itself. 
We will consider the binomial case computation before further generalizing the computation to the multinomial case.

\paragraph{The Binomial Case:}

In the Binomial case, terms are grouped together on the basis of being in the high $m_0$ state, or low $2-m_0$ state. Each state is therefore represented by a scalar occupation number which counts how many elements are in the high state: $n$.
We can therefore see that there are $k\choose n$ states. In each of these states with $n$ high states, and $n-k$ low states, there are $n!$ permutations which can be made among the high state elements, and $(n-k)!$ permutations among the low state elements
which leave the occupation number $n$ invariant. As such, the stabilizer is given by:
\begin{equation}
    \sum_{s \in \Sigma} |Stab(s)| = \sum_{n=0}^k {k\choose n} n!(k-n)! = \sum_{n=0}^k k! = (k+1)k! = (k+1)!
\end{equation}
plugging this back into \ref{Burnside} we get:
\begin{equation}
    |\Sigma / \mathcal{S}_k| = (k+1)
\end{equation}
An extra factor of $k$ is introduced due to the fact that the leave-one-out distrribution runs over $k$ components leading to a leading term complexity of $\mathcal{O}(k^2)$.

\paragraph{The Multinomial Case:}
In the Multinomial case, terms in the structure are grouped together on the basis of having a particular occupation configuration, as labeled by the occupation vector $\vec{n}$. We recall that the components of this vector are subject to the constraint $\sum_j n_j = k$. We can therefore see that there are ${k \choose n_1, \dots, n_D} = \frac{k!}{n_1!\dots n_D!}$ states. For each occupation vector configuration there are $\prod_j^D n_j!$ permutations among each occupation level. As a result, the summation is then:
\begin{equation}
    \sum_{s \in \Sigma} |Stab(s)| = \sum_{\vec{n},|\vec{n}| = k} \frac{k!}{\prod_j^D n_j!} \prod_j^D n_j! = \sum_{\vec{n},|\vec{n}| = k} k! = k! {k + D-1 \choose D-1}
\end{equation}
Where the combinatorial factor arises from the fact that we are arranging $k$ elements and $D-1$ bins into $D-1$ possible entries in the occupation vector $\vec{n}$.
We can then plug this back into \ref{Burnside} to get the size of the Orbit space:
\begin{equation}
    |\Sigma / \mathcal{S}_k| ={k + D-1 \choose D-1} = \frac{k^{D-1} + \dots+\prod_i^{D-1} i}{(D-1)!}
\end{equation}
As in the binomial case, an extra factor of $k$ is introduced by the leave-one-out distribution running over the $k$ components. This then leads to a leading dimensionality term of $O(k^D)$ under the inference algorithm, a polynomial time complexity.

\section{Implementation and Numerical Results}

In this section, we present the algorithmic formulation and implementation of the Sector Filter, and associated benchmark results. We wish to confirm three main claims for which three separate benchmarks were run. The first claim is that the Sector filter equals the naive filter in the regime where the naive filter is tractable. The second claim is that we wish to confirm the theoretical prediction of the $O(k^D)$ reduced time-complexity. The final claim which we wish to confirm is that the Filter is indeed exact, and yields logarithmic-likelihoods matching the Naive filter to acceptable precision where the latter is tractable.

\subsection{Sector Filter Algorithmic Implementation}

The theoretical construction developed in Sections 3 and 4 yields an exact filtering procedure operating on occupation-number configurations rather than the full state space. The implementation consists of three stages performed at each timestep: (i) prediction of the component marginals, (ii) construction of the predicted sector distribution, and (iii) Bayesian updating using the sector likelihoods.

Let
\[
\mathbf p_t = \{p_{\ell j,t}\}_{\ell=1,\ldots,k}^{j=1,\ldots,D}
\]
denote the collection of component-wise marginal probabilities and let
\[
q_t(\mathbf n)
\]
denote the sector distribution indexed by occupation vectors
\[
\mathbf n = (n_1,\ldots,n_D), \qquad \sum_{j=1}^{D} n_j = k.
\]

Rather than storing all \(D^k\) hidden states, the algorithm stores only the admissible occupation configurations, whose number is
\[
\binom{k+D-1}{D-1}.
\]

For each timestep \(t\), the prediction step first propagates the component marginals according to the factorized transition kernel \cite{augustyniak2026mdsv, ghahramani1997factorial}
\[
\tilde p_{\ell,t} = T_\ell p_{\ell,t-1},
\]
where \(T_\ell\) is the \(D \times D\) transition matrix defined in Eq. (20). Because the transition kernel factorizes over components, all predicted marginals can be computed independently.

The predicted sector distribution \(\tilde q_t(\mathbf n)\) is then constructed from the Poisson-multinomial generating function of Eq. (24). The associated leave-one-out distributions \(Q(\mathbf n,\ell)\) are computed recursively using Eq. (27) together with the boundary condition Eq. (28). These quantities contain all information required to recover the component marginals from the reduced occupation representation.

Given an observation \(x_t\), the likelihood associated with occupation vector \(\mathbf n\) is
\[
\omega_{\mathbf n}(x_t),
\]
which depends only on the occupation counts through the symmetric product structure established in Proposition 2. The sector posterior is obtained via Bayes' rule,
\[
q_t(\mathbf n)
=
\frac{\omega_{\mathbf n}(x_t)\tilde q_t(\mathbf n)}
{\sum_{\mathbf n'}
\omega_{\mathbf n'}(x_t)\tilde q_t(\mathbf n')}.
\]

Finally, the updated component marginals are recovered exactly from the sector posterior using the marginal recovery formula derived in Eq. (34),
\[
\tilde p_{\ell j,t}
=
\tilde p_{\ell j,t}
\frac{
\sum_{\mathbf n}
Q(\mathbf n-\hat e_j,\ell)\,
\omega_{\mathbf n}(x_t)
}{
\sum_{\mathbf n}
\omega_{\mathbf n}(x_t)\tilde q_t(\mathbf n)
}.
\]

The complete filtering procedure is summarized in Algorithm 1.

\textbf{Algorithm 1: Exact Sector Filter}

\begin{enumerate}
\item Initialize component marginals \(p_{\ell j,0}\).
\item For \(t=1,\ldots,T\):
\begin{enumerate}
\item Predict component marginals using Eq. (20).
\item Construct predicted sector probabilities \(\tilde q_t(\mathbf n)\).
\item Compute leave-one-out distributions \(Q(\mathbf n,\ell)\).
\item Evaluate sector likelihoods \(\omega_{\mathbf n}(x_t)\).
\item Update sector posterior \(q_t(\mathbf n)\).
\item Recover component marginals using Eq. (34).
\end{enumerate}
\item Accumulate the normalization constants to obtain the log-likelihood.
\end{enumerate}

Because all operations are performed on occupation-number configurations rather than the original state space, the implementation avoids the exponential \(D^k\) growth of the naive filter. The resulting computational cost is polynomial in the number of volatility components and agrees with the complexity analysis derived in Section 4. A lightweight Python Implementation of the algorithm can be accessed here: \url{https://github.com/DAHameedi32/msm_filter}.

\subsection{Time Complexity Benchmark}

For the Time Complexity benchmark, the Sector Filter was run on $T=2512$ days of GSPC data from the S$\&$P 500, fetched using the yfinance python library, with the underlying mathematical processes being handled by the NumPy library\cite{harris2020array}. Filtering was run against an array of k-values: $k \in \{5, 8, 10, 12, 15, 20, 25, 30, 40\}$ to demonstrate the relation for large k which has never been done attempted in the prior literature for computational reasons \cite{LUX202269}.We ran the filter for a variety of state-space sizes: $D=2,3,4$.
\newpage
\begin{figure}[h!]
    \centering
    \includegraphics[width=0.8\textwidth]{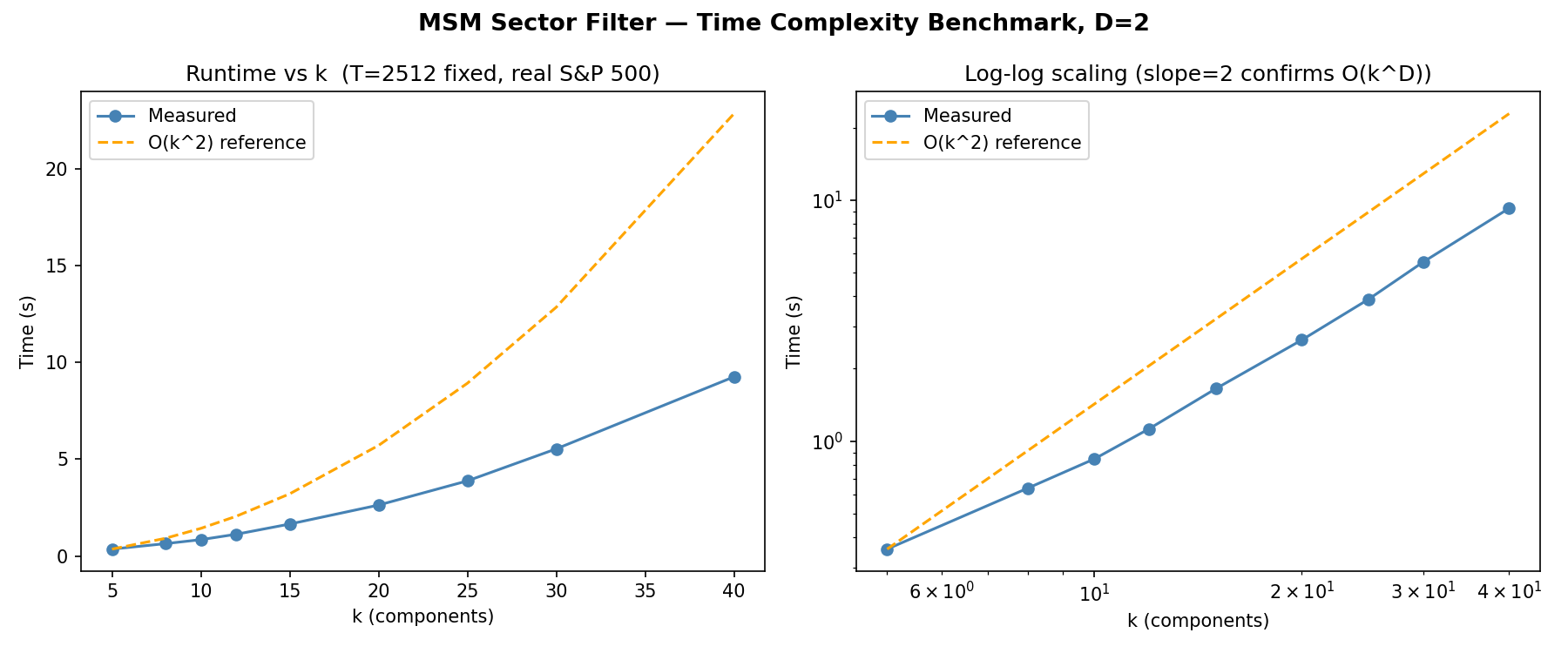}
    \caption{Time-scaling against theoretical prediction for the case $D=2$}
    \label{fig:D=2}
    \end{figure}
\begin{table}[h!]
    \begin{center}
    \begin{tabular}{||p{3cm} p{3cm} p{3cm} p{3cm}||} 
 \hline
 k & Time(s) & Ratio & $Tk^2$ \\ [0.5ex] 
 \hline\hline
     5  &    0.3572  &    1.00   &    62800\\
     8  &    0.6391  &    1.79   &   160768\\
    10  &    0.8454  &    1.32   &   251200\\
    12  &    1.1264  &    1.33   &   361728\\
    15  &    1.6523  &    1.47   &   565200\\
    20  &    2.6323  &    1.59   &  1004800\\
    25  &    3.8834  &    1.48   &  1570000\\
    30  &    5.5377  &    1.43  &   2260800\\
    40  &    9.2598  &    1.67 &    4019200\\ [1ex] 
 \hline
\end{tabular}
\end{center}
\caption{Raw time data for benchmark at $D=2, \{m_1 = 0.7, m_2 =1.3\}$. Ratio column indicates the ratio between that $k$ value's associated time and the previous $k$ value's.}

    \label{table:2}
\end{table}

\begin{figure}[h!]
    \centering
    \includegraphics[width=0.8\textwidth]{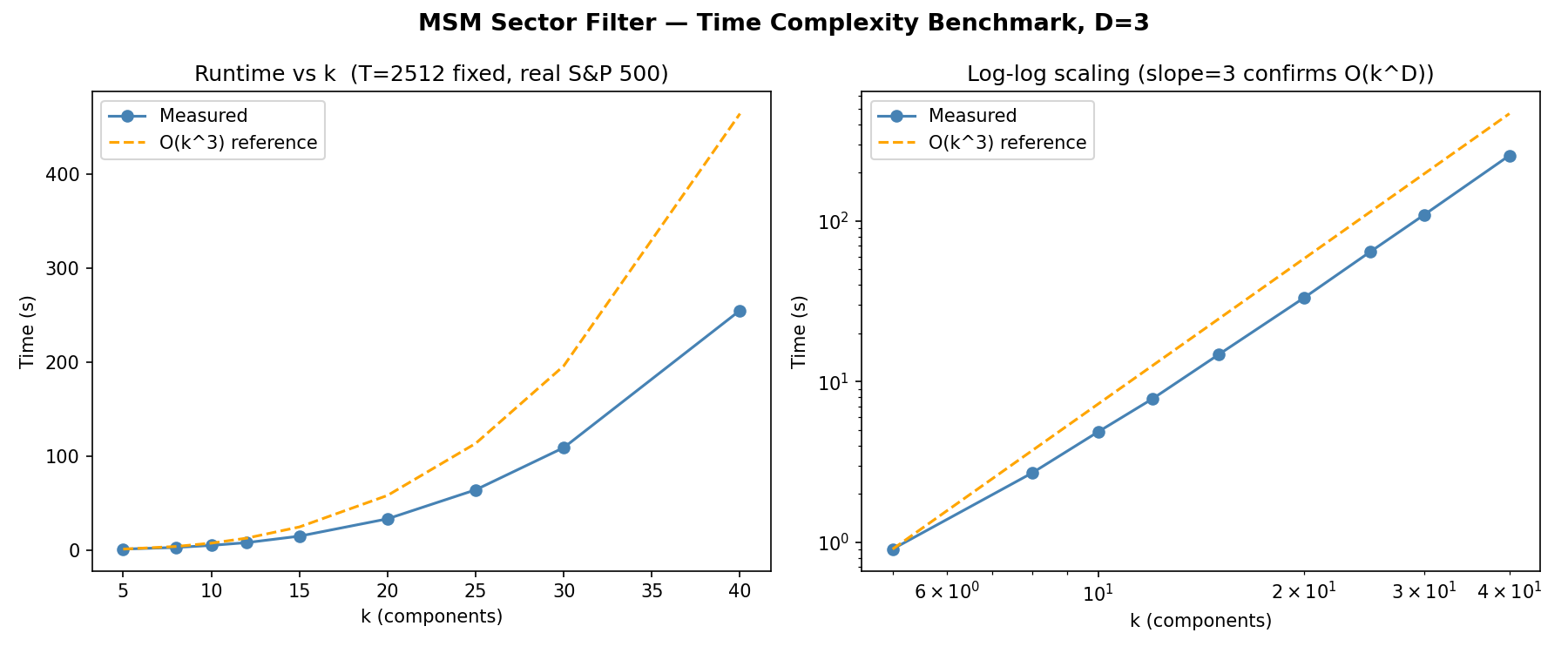}
    \caption{Time-scaling against theoretical prediction for the case $D=3$}
    \label{fig:D=3}
\end{figure}
\begin{table}[h!]
    \begin{center}
    \begin{tabular}{||p{3cm} p{3cm} p{3cm} p{3cm}||} 
 \hline
 k & Time(s) & Ratio & $Tk^3$ \\ [0.5ex] 
 \hline\hline
     5&      0.9073&      1.00&      314000\\
     8&      2.6984&      2.97&     1286144\\
    10&      4.8702&      1.80&     2512000\\
    12&      7.7983&      1.60&     4340736\\
    15&     14.7030&      1.89&     8478000\\
    20&     33.1357&      2.25&    20096000\\
    25&     64.0308&      1.93&    39250000\\
    30&    108.9450&      1.70&    67824000\\
    40&    255.0306&      2.34&   160768000\\[1ex] 
 \hline
\end{tabular}
\end{center}
\caption{Raw time data for benchmark at $D=3,  \{m_1 = 0.7, m_2 =1.0, m_3 =1.3\}$. Ratio column indicates the ratio between that $k$ value's associated time and the previous $k$ value's.}

    \label{table:3}
\end{table}
\begin{figure}[h!]
    \centering
    \includegraphics[width=0.8\textwidth]{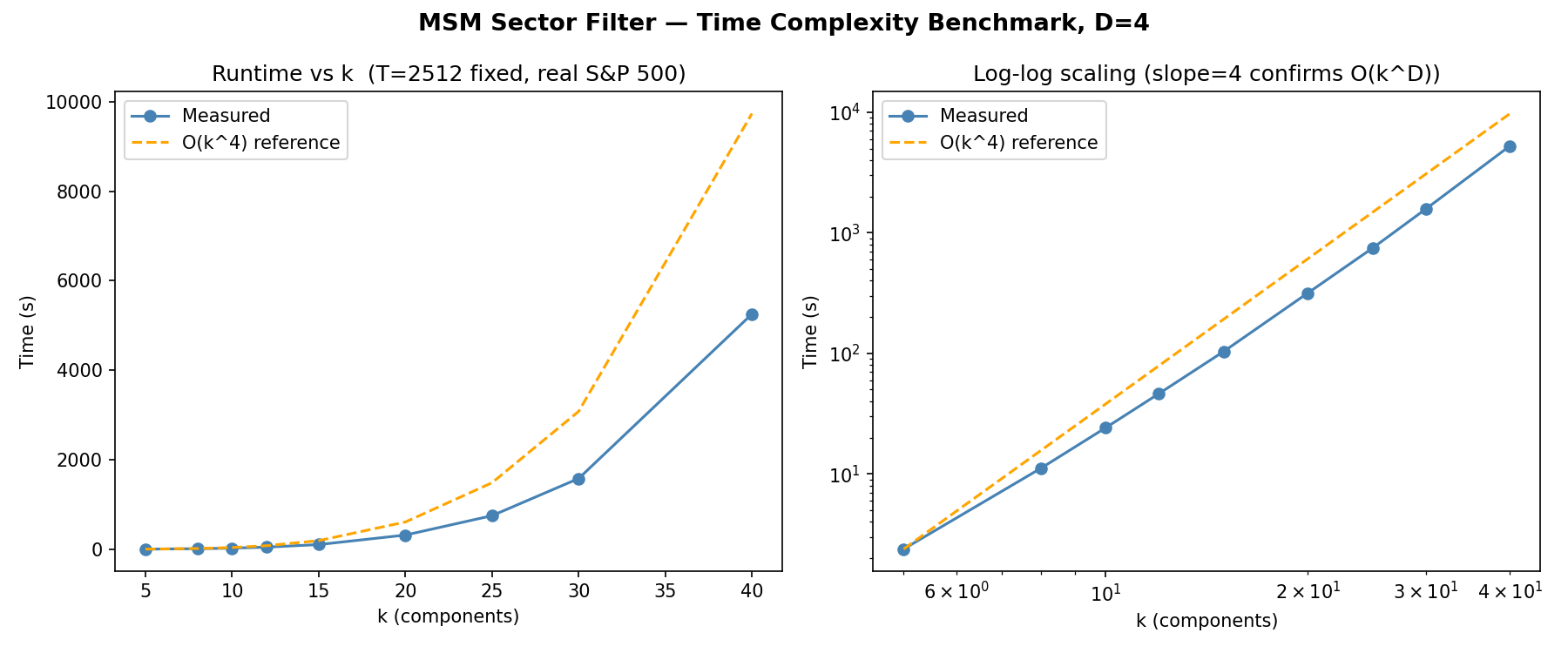}
    \caption{Time-scaling against theoretical prediction for the case $D=4$}
    \label{fig:D=3}
\end{figure}
\begin{table}[h!]
    \begin{center}
    \begin{tabular}{||p{3cm} p{3cm} p{3cm} p{3cm}||} 
 \hline
 k & Time(s) & Ratio & $Tk^4$ \\ [0.5ex] 
 \hline\hline
    5 &     2.3768 &     1.00 &    1570000\\
     8&     11.0946&      4.67&    10289152\\
    10&     24.1323&      2.18&    25120000\\
    12&     46.2650&      1.92&    52088832\\
    15&    103.9030&      2.25&   127170000\\
    20&    316.3284&      3.04&   401920000\\
    25&    749.3960&      2.37&   981250000\\
    30&   1580.3953&      2.11&  2034720000\\
    40&   5255.0459&      3.33&  6430720000\\[1ex] 
 \hline
\end{tabular}
\end{center}
\caption{Raw time data for benchmark at $D=4,  \{m_1 = 0.7, m_2 =0.9, m_3 =1.1, m_4=1.3\}$. Ratio column indicates the ratio between that $k$ value's associated time and the previous $k$ value's.}

    \label{table:4}
\end{table}
\newpage
We can see that compared to the $O(k^D)$ prediction, the filtering algorithm has over-performed at lower values of $D$. This is due to a variety of reasons, the main one being the fact that the NumPy library used to handle the underlying mathematical operations has a number of pre-packaged optimizations, such as running compiled C code, and also uses the BLAS and LAPACK libraries to self-optimize at run-time, and that it selects such libraries in order of performance from most to least\cite{numpy_blas_lapack}.

\subsection{Comparative Benchmark}
\label{sec:comparative-benchmark}
In this section we aim to empirically show that the Sector Filter operates in a numerically comparable manner to the naive filter in the regime where the naive filter is computationally feasible to compute. In this case we run for $N=20$ trials at $k=\{2,3,4,5,6\}$ over a process with $T=2512$ timesteps. This benchmark was run using a synthetic process taking the following set of multiplier values: $m = \{0.7, 0.9,1.0,1.1,1.3\}$, with an arbitrarily chosen non-ergodic initial distribution: $p_m = \{0.1, 0.15, 0.3, 0.25, 0.2\} $, and a pre-determined random seed generator: $\text{SEED}=42$. A summary comparison of the mean, absolute difference in log-likelihoods, as well as Hamming Distance, and MAP agreement of the two filters can be seen in the table below.
\begin{table}[htbp]
\centering
\caption{Exactness Benchmark and Comparison Summary per $k$ (Aggregated over 20 Trials)}
\label{tab:exactness_benchmark}
\resizebox{\textwidth}{!}{%
\begin{tabular}{rcccccccrcc}
\toprule
$k$  & \text{Mean } $|\Delta \text{LL}|$ & \text{Std } $|\Delta \text{LL}|$ & \text{Mean Rel } $\Delta \text{LL}$ & \text{Mean } $L_1$ & \text{Max } $L_1$ & \text{MAP Agree} & \text{Mean Hamming} & \text{Sector Rec.} & \text{Naive Rec.} \\
\midrule
2  & 0.013522  & 0.010815 & 0.000002 & 0.003690 & 0.011360 & 1.00 & 0.00 & 0.500000 & 0.500000 \\
3  & 0.680350  & 0.701281 & 0.000085 & 0.106122 & 0.311523 & 0.80 & 0.25 & 0.383333 & 0.416667 \\
4  & 0.768044  & 0.921561 & 0.000096 & 0.221339 & 0.658529 & 0.45 & 0.80 & 0.325000 & 0.300000 \\
5  & 0.899184  & 0.645930 & 0.000116 & 0.234066 & 0.400787 & 0.30 & 1.05 & 0.360000 & 0.360000 \\
6  & 1.550141  & 0.991777 & 0.000197 & 0.273904 & 0.513606 & 0.25 & 1.25 & 0.258333 & 0.250000 \\
\bottomrule
\end{tabular}%
}
\end{table}
The full raw data for this benchmark can be found in the Appendix A. We can isolate some individual features of the Benchmark in the following graphs:
\begin{figure}[h!]
    \centering
    \includegraphics[width=0.8\textwidth]{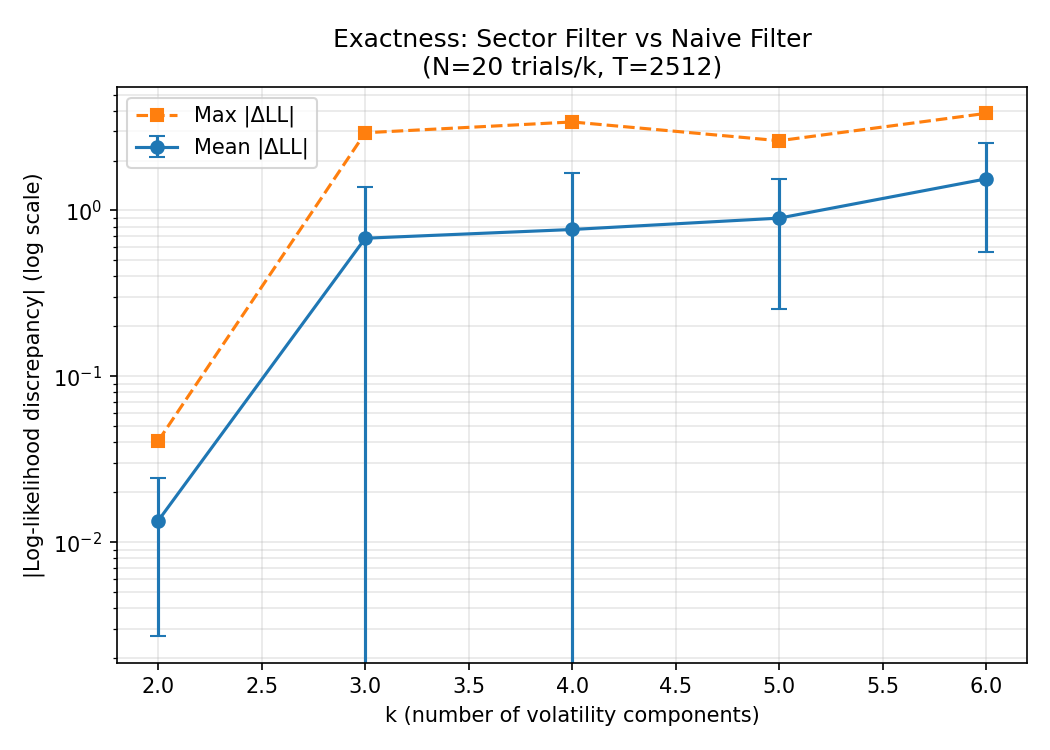}
    \caption{ $|\Delta \text{LL}| $ plotted against $k$, the number of multivariate volatility components}
    \label{fig:loglik_discrepancy}
\end{figure}
\begin{figure}[h!]
    \centering
    \includegraphics[width=0.8\textwidth]{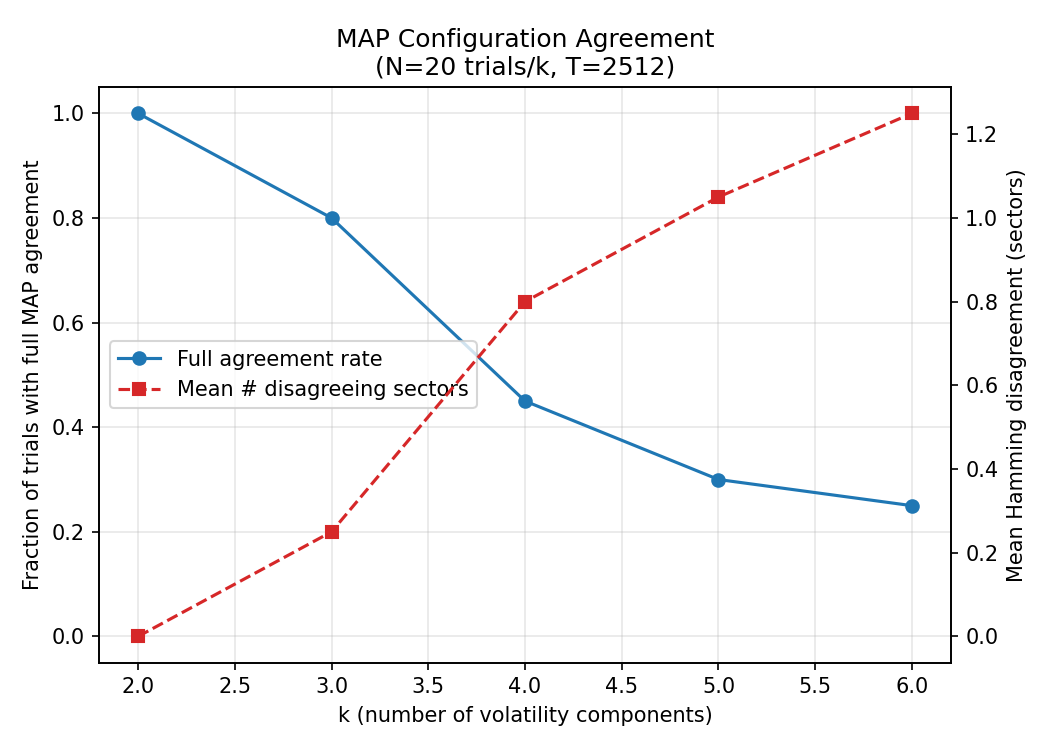}
    \caption{ MAP Agreement between Naive and Sector filters with corresponding Hamming disagreement}
    \label{fig:map_agreement}
\end{figure}
\begin{figure}[h!]
    \centering
    \includegraphics[width=0.8\textwidth]{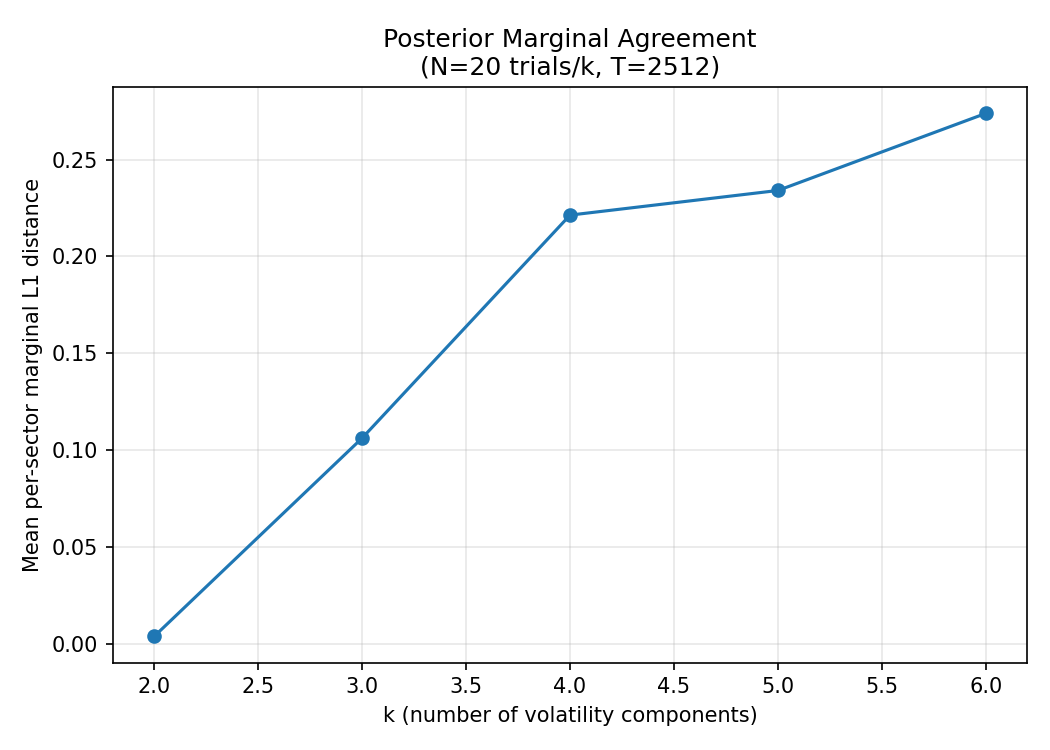}
    \caption{ $|\Delta \text{LL}| $ plotted against $k$, the number of multivariate volatility components}
    \label{fig:marginal_l1_distance}
\end{figure}

\begin{figure}[h!]
    \centering
    \includegraphics[width=0.8\textwidth]{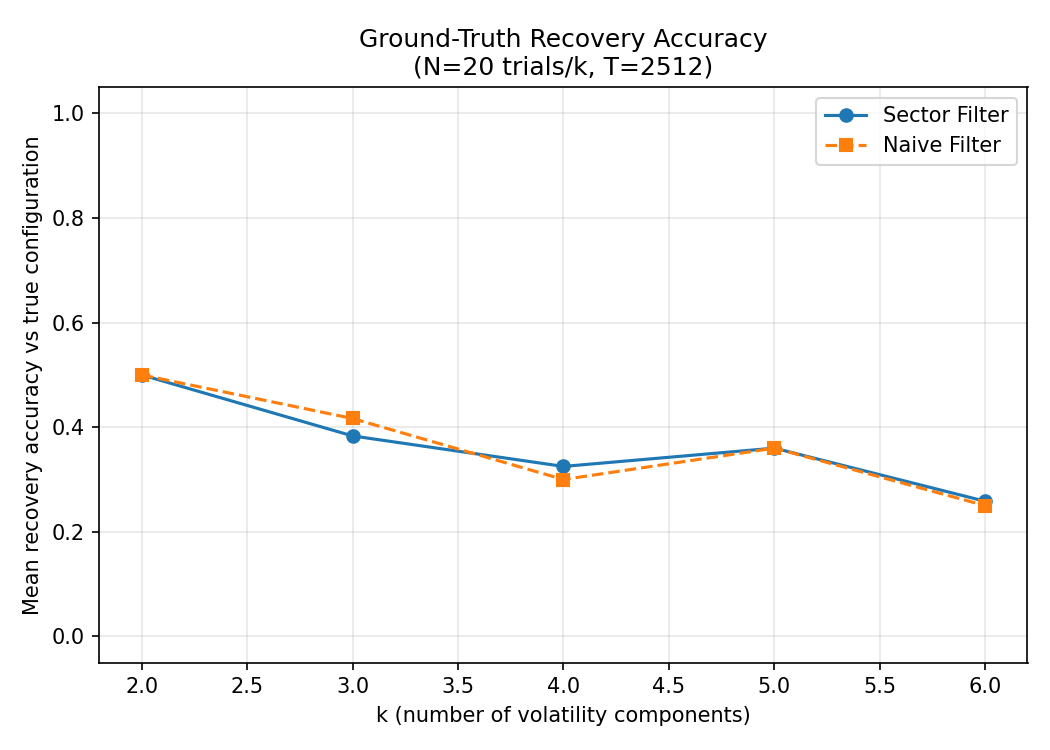}
    \caption{ Mean number of correct components as a fraction of the ground truth of the Synthetic MMSM process}
    \label{fig:lrecovery_accuracy}
\end{figure}
Finally we can see that the mean recovery accuracy of the two filters are directly comparable in the tractable regime of the naive filter.
\subsection{Exactness Benchmark}
\label{sec:exactness-benchmark}
To evaluate the accuracy of the sector filter, we compare its outputs against those of the full-state Hidden Markov Model filter on synthetic datasets generated from the same model. For each value of ($k\in{2,3,4,5,6}$), twenty independent simulations were performed and both filters were run on identical observations.

The comparison uses five metrics: (i) the log-likelihood difference,
\begin{equation}
   |  \mathrm{LL}_{\mathrm{Sector}} - \mathrm{LL}_{\mathrm{Naive}}|
\end{equation}

(ii) the relative log-likelihood difference, (iii) the mean ($L_1$) distance between recovered component marginals, (iv) the Hamming distance between maximum-a-posteriori (MAP) state estimates, and (v) agreement rates for sector-level state recovery.

Table 4 summarizes the results. For (k=2), the sector filter is numerically indistinguishable from the full-state filter, with relative log-likelihood differences on the order of ($10^{-6}$), negligible marginal discrepancies, and perfect MAP agreement across all trials. As (k) increases, deviations emerge. The mean ($L_1$) error grows gradually from approximately (0.005) at (k=2) to approximately (0.29) at (k=6), while the relative log-likelihood discrepancy remains below ($5\times10^{-4}$) in all experiments.

Despite these differences, the inferred states remain highly consistent between the two filters. MAP agreement is perfect for (k=2) but drops off for larger values of (k), with disagreements occurring primarily in trajectories where a large and increasing number of hidden-state configurations possess nearly identical posterior probability. 

We can attribute this to the inherent multiplicity of the sector-filter posterior distributions, where multiple configuration marginals recovered from the lifted representation have similar posterior likelihoods, a conclusion partially supported by the close agreement in recovery accuracy between the two filters. Additional concerns for the rising disagreement may arise from the accumulation of gradual numerical floating point precision errors in the computation of $Q_{(n, \ell)}$ involving division by $\tilde{p}_{\ell, D}$, though a full numerical stability study is not carried out here.

These results indicate that the sector filter preserves the dominant probabilistic structure of the full-state model while operating on a substantially reduced state space. Even for the largest benchmark considered, discrepancies in likelihood and state recovery remain small relative to the scale of the filtering problem.

\subsection{Identifiability Benchmark}
\label{sec:identifiability}

The discrepancies reported in Sections~\ref{sec:comparative-benchmark} and~\ref{sec:exactness-benchmark} grow with $k$, which raises the question of whether they reflect genuine divergence between the sector and naive filters, or an artefact of comparing labelled full-state configurations under a likelihood that is, by construction (Proposition~2), invariant under $S_k$. Two configurations $x, y \in \Sigma$ lying in the same orbit --- i.e.\ $y = \sigma \cdot x$ for some $\sigma \in S_k$ --- are assigned identical likelihood and hence identical posterior mass by both filters; any metric that penalises disagreement between such configurations is measuring a labelling artefact rather than a filtering error. We therefore introduce an orbit-invariant analogue of each metric used in Sections~\ref{sec:comparative-benchmark}--\ref{sec:exactness-benchmark} and re-run the comparative benchmark under it.

\paragraph{Orbit-invariant distance.} Define
\begin{equation}
d_{\mathrm{orbit}}(x,y) := \min_{\sigma\in S_k} d_H(x,\sigma\cdot y),
\label{eq:orbit-dist-def}
\end{equation}
the Hamming distance between $x$ and $y$ minimised over all relabellings of $y$. This has a closed form in terms of the occupation vectors $n(x), n(y) \in \mathbb{Z}_{\geq 0}^D$ introduced in Section~4, avoiding any search over $S_k$:

\begin{lemma}
For $x, y \in \Sigma$ with occupation vectors $n(x), n(y)$,
\begin{equation}
d_{\mathrm{orbit}}(x,y) = \frac{1}{2}\lVert n(x)-n(y)\rVert_1.
\label{eq:orbit-dist-closed-form}
\end{equation}
\end{lemma}

\begin{proof}
A relabelling $\sigma$ can match at most $\min(n_j(x), n_j(y))$ of the components in state $j$ between the two configurations, for each $j = 1,\dots,D$. Summing over $j$ gives the maximum achievable number of agreeing positions, so
\[
k - d_{\mathrm{orbit}}(x,y) = \sum_{j=1}^{D} \min\big(n_j(x), n_j(y)\big).
\]
Using the identity $\sum_j \min(a_j,b_j) = \tfrac{1}{2}\big(\sum_j a_j + \sum_j b_j - \sum_j |a_j - b_j|\big)$, together with $\sum_j n_j(x) = \sum_j n_j(y) = k$, gives Eq.~\eqref{eq:orbit-dist-closed-form}. The bound is achievable: greedily pairing the $\min(n_j(x),n_j(y))$ overlapping components in each state $j$, then assigning an arbitrary bijection between the remaining unmatched components, realises a permutation $\sigma$ attaining it.
\end{proof}

\paragraph{Orbit-invariant MAP agreement.} In place of the raw Hamming distance and MAP agreement indicator reported in Table~4 and Figure~5, we report
\begin{equation}
d_{\mathrm{orbit}}\big(x^*_{\mathrm{sector}}, x^*_{\mathrm{naive}}\big) = \frac{1}{2}\big\lVert n(x^*_{\mathrm{sector}}) - n(x^*_{\mathrm{naive}})\big\rVert_1,
\label{eq:orbit-map}
\end{equation}
where $x^*_{\mathrm{sector}}$ and $x^*_{\mathrm{naive}}$ are the MAP full-state configurations recovered by each filter, together with the orbit-level MAP agreement indicator $\mathbf{1}\big[n(x^*_{\mathrm{sector}}) = n(x^*_{\mathrm{naive}})\big]$.

\paragraph{Orbit-invariant recovery accuracy.} In place of the per-component correctness fraction plotted in Figure~7, we report
\begin{equation}
\mathrm{Acc}_{\mathrm{orbit}}(\hat x) := \frac{k - d_{\mathrm{orbit}}(x_{\mathrm{true}}, \hat x)}{k},
\label{eq:orbit-acc}
\end{equation}
the best achievable fraction of components correctly recovered under an optimal relabelling of the estimate $\hat x$ against the ground-truth configuration $x_{\mathrm{true}}$, computed separately for $\hat x = x^*_{\mathrm{sector}}$ and $\hat x = x^*_{\mathrm{naive}}$.

\paragraph{Orbit-level posterior agreement.} The marginal $L_1$ metric in Table~4 compares per-component marginals $p_{\ell j,t}$, which are themselves only identifiable up to relabelling and so inherit the same artefact. We replace it with a direct comparison of the sector posterior $q_t(n)$, computed identically by both filters since it is by definition the pushforward of the full posterior along the orbit map $\pi:\Sigma\to\Sigma/S_k,\ x\mapsto n(x)$:
\begin{equation}
\mathrm{TV}\big(q^{\mathrm{sector}}_t, q^{\mathrm{naive}}_t\big) = \frac{1}{2}\sum_{n} \left| q^{\mathrm{sector}}_t(n) - q^{\mathrm{naive}}_t(n)\right|,
\label{eq:orbit-tv}
\end{equation}
where $q^{\mathrm{naive}}_t(n) := \sum_{x:\,\pi(x)=n} p_{\mathrm{naive}}(x\mid y_{1:t})$ is obtained by marginalising the naive filter's full-state posterior onto occupation vectors. Because this compares two distributions on the same orbit space $\Sigma/S_k$ rather than labelled point estimates, it is orbit-invariant by construction and needs no correction of its own.

\paragraph{Benchmark protocol.} We repeat the comparative benchmark of Section~\ref{sec:comparative-benchmark} under identical settings ($N=20$ trials per $k \in \{2,3,4,5,6\}$, $T=2512$, multiplier values $m=\{0.7,0.9,1.0,1.1,1.3\}$, seed $=42$), replacing each of the three metrics above with its orbit-invariant counterpart from Eqs.~\eqref{eq:orbit-map}--\eqref{eq:orbit-tv}. Table~\ref{tab:identifiability} reports the results.

\begin{table}[h]
\centering
\begin{tabular}{c c c c c c c}
\hline
$k$ & Mean $d_{\mathrm{orbit}}$ & Std $d_{\mathrm{orbit}}$ & Orbit MAP Agree & Mean TV$(q^{\mathrm{sec}},q^{\mathrm{nai}})$ & Sector $\mathrm{Acc}_{\mathrm{orbit}}$ & Naive $\mathrm{Acc}_{\mathrm{orbit}}$ \\
\hline
2 & 0.2000 & 0.4104 & 0.80 & 0.0024 & 0.6500 & 0.6750 \\
3 & 0.6000 & 0.5026 & 0.40 & 0.0745 & 0.6500 & 0.5333 \\
4 & 1.5000 & 0.5130 & 0.00 & 0.1639 & 0.6750 & 0.5500 \\
5 & 1.5000 & 0.5130 & 0.00 & 0.1688 & 0.6900 & 0.5600 \\
6 & 2.5500 & 0.6048 & 0.00 & 0.2753 & 0.6917 & 0.4583 \\
\hline
\end{tabular}
\caption{Identifiability Benchmark and Comparison Summary per $k$ (Aggregated over 20 Trials). Compare against Table~4.}
\label{tab:identifiability}
\end{table}

The full orbit-invariant benchmark data output can be seen in Appendix B.

We consider now that these final set of results seem to throw the trend observed in the cases of the configuration posterior level Exactness and Comparative benchmarks into a sharper relief. The sector filter frequently seems to equal or outperform the naive filter regarding ground-truth recovery of volatility component configuration, yet they disagree significantly on structure. The major source of discrepancy between the two filtering procedures arises primarily in the MAP agreement between the two, which largely arises from the inherent identifiability differences between the two methodologies.
MAP agreement should not be used as the primary measure of filter performance here because it measures estimator agreement rather than recovery of the latent state. This distinction is particularly important for the MSM considered here, whose microscopic state space possesses permutation symmetry. Distinct microscopic configurations may therefore correspond to the same orbit and are indistinguishable at the sector level. Moreover, the MAP microscopic configuration need not belong to the orbit carrying the greatest aggregate posterior probability. It is therefore better to evaluate recovery against the known generating state, both at the microscopic and orbit-invariant levels.
\section{Conclusions and Outlook}

In this paper we formulated a reduced filtering procedure for the discrete multivariate Markov-switching multifractal (MSM) model by exploiting two structural properties of the model: the permutation symmetry of its likelihood and the factorization of its transition kernel. This permits the labelled $D^k$ configuration space to be replaced by an occupation-vector representation, producing substantial reductions in the computational complexity of MSM filtering. The resulting construction suggests a broader principle: when the likelihood is invariant under a group action and the transition operator respects the corresponding symmetry, filtering may be formulated directly on the associated orbit or sector space rather than on the redundant microscopic state space.

Several questions remain open. In particular, the final benchmarking results motivate a more careful distinction between filter agreement, distributional agreement, and ground-truth recovery. The MAP agreement criterion used in parts of this study is insufficient as a standalone measure of filter performance. Agreement of MAP configurations measures agreement between two estimators rather than correctness with respect to the latent state, and is particularly problematic in the present setting because permutation-related configurations are intentionally identified by the sector construction. In general,
\begin{equation}
q\!\left(\arg\max_x \pi(x)\right)
\neq
\arg\max_s \sum_{x:q(x)=s} \pi(x),
\end{equation}
so microscopic MAP agreement need not correspond even to agreement at the level of the most probable orbit.

The ground-truth recovery experiments provide a complementary and potentially more informative observation. Across the benchmarked regimes, the sector filter frequently exhibits higher recovery accuracy than the naive filter, both for labelled configurations and for orbit-invariant recovery. This occurs even in cases where the MAP configurations selected by the two filters disagree. Consequently, disagreement between the two posterior representations cannot by itself be interpreted as evidence that the sector filter is performing worse. Instead, the observed recovery advantage raises a further question: \emph{why should explicitly enforcing the permutation-invariant combinatorial structure of the MSM improve latent-state recovery?}

One possible explanation is that the sector representation acts as an implicit regularisation of the inference problem. The naive filter distinguishes between
\begin{equation}
\frac{k!}{n_1!\cdots n_D!}
\end{equation}
labelled configurations belonging to the same occupation sector, despite the likelihood being invariant under permutations of those labels. The sector filter instead aggregates these configurations into their minimal combinatorial representation before inference. This may suppress permutation-level posterior fragmentation, reduce variance associated with redundant microscopic states, and provide a more stable estimate of the collective volatility variable to which the observation process is sensitive. Other explanations, including the geometry of the MSM posterior and the possibility that the sector representation acts as a form of symmetry-induced regularisation, remain possible. Distinguishing between these mechanisms requires further controlled experiments.

The benchmarking protocol itself can therefore be strengthened in several directions. Future studies should separately report microscopic recovery accuracy, orbit-invariant recovery accuracy, distributional agreement between the sector posterior and the orbit projection of the naive posterior, and computational complexity. In particular, the quantity
\begin{equation}
D_{\mathrm{TV}}
\left(
q_{\#}\pi_t^{\mathrm{naive}},
\pi_t^{\mathrm{sector}}
\right)
\end{equation}
provides a direct measure of agreement on the quotient state space, while paired estimates of
\begin{equation}
\Delta A
=
A_{\mathrm{sector}}-A_{\mathrm{naive}}
\end{equation}
provide a direct measure of relative ground-truth recovery. Repeated matched simulations and uncertainty estimates for these quantities would allow the apparent recovery advantage observed here to be tested systematically.

A second direction concerns generalisations to interacting populations. If $n$ populations of sizes $k_1,\ldots,k_n$ are each internally exchangeable, the permutation symmetry generalises naturally to
\begin{equation}
G =
S_{k_1}\times S_{k_2}\times\cdots\times S_{k_n}.
\end{equation}
The corresponding state representation would consist of one occupation vector per population, with a sector-space cardinality of
\begin{equation}
\prod_{a=1}^{n}
\binom{k_a+D-1}{D-1}.
\end{equation}
This provides a natural extension of the present reduction to models containing multiple interacting but internally exchangeable populations. The more difficult case occurs when interactions break exchangeability, in which case the appropriate residual symmetry and corresponding orbit representation must first be identified.

Finally, a more speculative direction is the extension from discrete permutation symmetries to continuous symmetry groups in continuum formulations of MSM-like models. Rather than treating $S_k$ as literally becoming $U(k)$ or $SU(k)$, one may ask whether a continuum latent multiplier model possesses a continuous group $G$ acting on its state manifold $M$. The discrete occupation-space construction would then have the analogue
\begin{equation}
M \longrightarrow M/G,
\end{equation}
with the group orbits replacing discrete permutation orbits. For a Lie group, the corresponding finite-dimensional Lie algebra of infinitesimal generators provides a natural language for characterising these orbits and their invariants. Whether such a continuous symmetry reduction can produce a computationally finite filtering procedure, and under what conditions, remains an open problem. Nevertheless, it suggests that the symmetry-based reduction developed here may represent a special case of a broader approach to filtering on quotient state spaces.

Overall, the present results establish the computational value of exploiting permutation structure while also raising a more subtle question concerning the relationship between symmetry reduction, posterior geometry, and latent-state recovery. Resolving that question, together with establishing rigorous conditions for exact quotient filtering and extending the construction to more general symmetry groups, provides a natural programme for future work.

\printbibliography
\newpage
\appendix
\section{Appendix A: Raw Benchmark Data}
\small
\begin{longtable}{rcccccccccc}
\caption{Complete Raw Per-Trial Benchmark Results ($T=2512$ timesteps, 20 trials per $k$)} \label{tab:benchmark_raw_data} \\
\toprule
$k$ & Trial & $\text{LL}_{\text{Sector}}$ & $\text{LL}_{\text{Naive}}$ & $|\Delta \text{LL}|$ & Rel $\Delta \text{LL}$ & Mean $L_1$ & Hamming & MAP Agree & Sector Hit & Naive Hit \\
\midrule
\endfirsthead

\multicolumn{11}{c}{{\bfseries \tablename\ \thetable{} -- continued from previous page}} \\
\toprule
$k$ & Trial & $\text{LL}_{\text{Sector}}$ & $\text{LL}_{\text{Naive}}$ & $|\Delta \text{LL}|$ & Rel $\Delta \text{LL}$ & Mean $L_1$ & Hamming & MAP Agree & Sector Hit & Naive Hit \\
\midrule
\endhead

\midrule 
\multicolumn{11}{r}{{Continued on next page}} \\ 
\bottomrule
\endfoot

\bottomrule
\endlastfoot

2 & 0 & 7901.2214 & 7901.2222 & 0.000784 & 9.92e-08 & 0.004908 & 0 & True & 1.000 & 1.000 \\
2 & 1 & 7862.4658 & 7862.4786 & 0.012799 & 1.63e-06 & 0.000729 & 0 & True & 0.500 & 0.500 \\
2 & 2 & 8405.6599 & 8405.6557 & 0.004143 & 4.93e-07 & 0.000197 & 0 & True & 0.500 & 0.500 \\
2 & 3 & 7885.0070 & 7885.0063 & 0.000689 & 8.74e-08 & 0.006969 & 0 & True & 0.000 & 0.000 \\
2 & 4 & 7786.0802 & 7786.0852 & 0.005017 & 6.44e-07 & 0.007758 & 0 & True & 0.500 & 0.500 \\
2 & 5 & 7990.7497 & 7990.7632 & 0.013457 & 1.68e-06 & 0.003665 & 0 & True & 0.000 & 0.000 \\
2 & 6 & 7822.4323 & 7822.4439 & 0.011587 & 1.48e-06 & 0.003956 & 0 & True & 1.000 & 1.000 \\
2 & 7 & 8194.4233 & 8194.4071 & 0.016253 & 1.98e-06 & 0.000151 & 0 & True & 0.500 & 0.500 \\
2 & 8 & 7576.1810 & 7576.1609 & 0.020073 & 2.65e-06 & 0.000205 & 0 & True & 1.000 & 1.000 \\
2 & 9 & 7618.9205 & 7618.9282 & 0.007680 & 1.01e-06 & 0.001012 & 0 & True & 1.000 & 1.000 \\
2 & 10 & 7919.1248 & 7919.1417 & 0.016866 & 2.13e-06 & 0.004052 & 0 & True & 0.000 & 0.000 \\
2 & 11 & 7847.6628 & 7847.6891 & 0.026315 & 3.35e-06 & 0.003984 & 0 & True & 0.500 & 0.500 \\
2 & 12 & 7927.8184 & 7927.8427 & 0.024319 & 3.07e-06 & 0.002875 & 0 & True & 0.500 & 0.500 \\
2 & 13 & 7928.3242 & 7928.3541 & 0.029863 & 3.77e-06 & 0.002958 & 0 & True & 0.500 & 0.500 \\
2 & 14 & 7730.0142 & 7730.0552 & 0.040941 & 5.30e-06 & 0.011360 & 0 & True & 0.500 & 0.500 \\
2 & 15 & 8056.2415 & 8056.2625 & 0.020993 & 2.61e-06 & 0.000632 & 0 & True & 0.500 & 0.500 \\
2 & 16 & 7864.0416 & 7864.0456 & 0.004033 & 5.13e-07 & 0.011340 & 0 & True & 0.500 & 0.500 \\
2 & 17 & 8140.4042 & 8140.4075 & 0.003299 & 4.05e-07 & 0.001691 & 0 & True & 0.000 & 0.000 \\
2 & 18 & 7961.4284 & 7961.4357 & 0.007328 & 9.20e-07 & 0.005166 & 0 & True & 0.500 & 0.500 \\
2 & 19 & 8043.9142 & 8043.9182 & 0.003975 & 4.94e-07 & 0.001881 & 0 & True & 0.000 & 0.000 \\
\midrule
3 & 0 & 7962.3813 & 7963.0906 & 0.709325 & 8.91e-05 & 0.093120 & 0 & True & 0.333 & 0.333 \\
3 & 1 & 7887.6713 & 7888.1969 & 0.525624 & 6.66e-05 & 0.071850 & 0 & True & 0.333 & 0.333 \\
3 & 2 & 8459.7744 & 8459.9575 & 0.183060 & 2.16e-05 & 0.038596 & 0 & True & 0.667 & 0.667 \\
3 & 3 & 7979.6267 & 7980.2078 & 0.581121 & 7.28e-05 & 0.093181 & 0 & True & 0.333 & 0.333 \\
3 & 4 & 7880.8931 & 7881.6508 & 0.757657 & 9.61e-05 & 0.111816 & 0 & True & 0.333 & 0.333 \\
3 & 5 & 8031.5794 & 8031.7645 & 0.185055 & 2.30e-05 & 0.033621 & 0 & True & 0.333 & 0.333 \\
3 & 6 & 7868.0485 & 7868.6186 & 0.570146 & 7.25e-05 & 0.094565 & 0 & True & 0.333 & 0.333 \\
3 & 7 & 8282.8805 & 8283.4735 & 0.592982 & 7.16e-05 & 0.084705 & 0 & True & 0.333 & 0.333 \\
3 & 8 & 7682.2039 & 7685.1449 & 2.940995 & 3.83e-04 & 0.311523 & 1 & False & 0.667 & 0.333 \\
3 & 9 & 7685.0084 & 7685.1558 & 0.147441 & 1.92e-05 & 0.037130 & 0 & True & 0.667 & 0.667 \\
3 & 10 & 7984.7766 & 7985.3409 & 0.564243 & 7.07e-05 & 0.091176 & 0 & True & 0.000 & 0.000 \\
3 & 11 & 7887.8961 & 7889.0306 & 1.134542 & 1.44e-04 & 0.158300 & 1 & False & 0.333 & 0.667 \\
3 & 12 & 7979.7431 & 7980.2520 & 0.508931 & 6.38e-05 & 0.081878 & 0 & True & 0.333 & 0.333 \\
3 & 13 & 7983.8596 & 7984.8143 & 0.954660 & 1.20e-04 & 0.138402 & 0 & True & 0.333 & 0.333 \\
3 & 14 & 7772.3394 & 7773.0640 & 0.724623 & 9.32e-05 & 0.103986 & 0 & True & 0.333 & 0.333 \\
3 & 15 & 8107.0398 & 8107.3916 & 0.351795 & 4.34e-05 & 0.055811 & 0 & True & 0.333 & 0.333 \\
3 & 16 & 7949.1993 & 7950.4131 & 1.213824 & 1.53e-04 & 0.170669 & 1 & False & 0.667 & 0.667 \\
3 & 17 & 8206.1824 & 8206.2753 & 0.092898 & 1.13e-05 & 0.021617 & 0 & True & 0.000 & 0.000 \\
3 & 18 & 8011.0857 & 8011.6917 & 0.606018 & 7.56e-05 & 0.088825 & 0 & True & 0.667 & 0.667 \\
3 & 19 & 8089.4719 & 8089.9472 & 0.475283 & 5.87e-05 & 0.078601 & 1 & False & 0.000 & 0.667 \\
\midrule
4 & 0 & 7986.7291 & 7987.5457 & 0.816601 & 1.02e-04 & 0.169727 & 0 & True & 0.250 & 0.250 \\
4 & 1 & 7922.3610 & 7923.3618 & 1.000780 & 1.26e-04 & 0.198335 & 1 & False & 0.250 & 0.250 \\
4 & 2 & 8489.1760 & 8489.6588 & 0.482810 & 5.69e-05 & 0.117178 & 0 & True & 0.500 & 0.500 \\
4 & 3 & 8012.7297 & 8013.6277 & 0.897972 & 1.12e-04 & 0.180429 & 1 & False & 0.250 & 0.250 \\
4 & 4 & 7919.6469 & 7920.6120 & 0.965042 & 1.22e-04 & 0.191286 & 1 & False & 0.250 & 0.250 \\
4 & 5 & 8056.4026 & 8057.2514 & 0.848777 & 1.05e-04 & 0.181822 & 1 & False & 0.250 & 0.000 \\
4 & 6 & 7894.2255 & 7895.0450 & 0.819485 & 1.04e-04 & 0.176435 & 0 & True & 0.250 & 0.250 \\
4 & 7 & 8314.9789 & 8315.8239 & 0.844962 & 1.02e-04 & 0.170660 & 0 & True & 0.250 & 0.250 \\
4 & 8 & 7731.3323 & 7734.7487 & 3.416342 & 4.42e-04 & 0.658529 & 2 & False & 0.500 & 0.500 \\
4 & 9 & 7720.6274 & 7720.9169 & 0.289456 & 3.75e-05 & 0.082729 & 0 & True & 0.500 & 0.500 \\
4 & 10 & 8007.4102 & 8008.2045 & 0.794354 & 9.92e-05 & 0.165180 & 0 & True & 0.000 & 0.000 \\
4 & 11 & 7914.5160 & 7915.0069 & 0.490890 & 6.20e-05 & 0.113063 & 1 & False & 0.250 & 0.000 \\
4 & 12 & 8007.8282 & 8008.6291 & 0.800912 & 1.00e-04 & 0.168913 & 0 & True & 0.250 & 0.250 \\
4 & 13 & 8011.6669 & 8011.8315 & 0.164532 & 2.05e-05 & 0.051915 & 0 & True & 0.250 & 0.250 \\
4 & 14 & 7799.3878 & 7800.2743 & 0.886561 & 1.14e-04 & 0.181827 & 1 & False & 0.250 & 0.250 \\
4 & 15 & 8137.6696 & 8138.2709 & 0.601323 & 7.39e-05 & 0.133284 & 1 & False & 0.250 & 0.250 \\
4 & 16 & 7977.1706 & 7977.5841 & 0.413524 & 5.18e-05 & 0.098495 & 1 & False & 0.750 & 0.500 \\
4 & 17 & 8234.3312 & 8234.6934 & 0.362242 & 4.40e-05 & 0.089853 & 0 & True & 0.000 & 0.000 \\
4 & 18 & 8043.9056 & 8044.5772 & 0.671626 & 8.35e-05 & 0.147137 & 2 & False & 0.500 & 0.500 \\
4 & 19 & 8118.8249 & 8119.5475 & 0.722699 & 8.90e-05 & 0.151952 & 1 & False & 0.000 & 0.250 \\
\midrule
5 & 0 & 8000.3478 & 8001.3256 & 0.977803 & 1.22e-04 & 0.243577 & 1 & False & 0.200 & 0.200 \\
5 & 1 & 7943.0807 & 7944.0202 & 0.939525 & 1.18e-04 & 0.233076 & 1 & False & 0.200 & 0.200 \\
5 & 2 & 8506.7578 & 8507.7289 & 0.971097 & 1.14e-04 & 0.234327 & 1 & False & 0.400 & 0.400 \\
5 & 3 & 8031.5791 & 8032.5539 & 0.974728 & 1.21e-04 & 0.239339 & 1 & False & 0.200 & 0.200 \\
5 & 4 & 7942.0620 & 7943.0645 & 1.002534 & 1.26e-04 & 0.248698 & 1 & False & 0.200 & 0.200 \\
5 & 5 & 8072.1818 & 8073.0722 & 0.890325 & 1.10e-04 & 0.222384 & 1 & False & 0.200 & 0.200 \\
5 & 6 & 7909.8329 & 7910.7410 & 0.908070 & 1.15e-04 & 0.225642 & 1 & False & 0.200 & 0.200 \\
5 & 7 & 8332.2575 & 8333.1951 & 0.937611 & 1.13e-04 & 0.227448 & 1 & False & 0.200 & 0.200 \\
5 & 8 & 7762.6289 & 7765.2622 & 2.633303 & 3.39e-04 & 0.400787 & 2 & False & 0.800 & 0.800 \\
5 & 9 & 7745.0644 & 7745.3371 & 0.272718 & 3.52e-05 & 0.101166 & 0 & True & 0.400 & 0.400 \\
5 & 10 & 8019.2079 & 8020.1581 & 0.950186 & 1.18e-04 & 0.236610 & 1 & False & 0.000 & 0.000 \\
5 & 11 & 7929.5638 & 7930.5593 & 0.995471 & 1.26e-04 & 0.245864 & 0 & True & 0.200 & 0.200 \\
5 & 12 & 8022.0305 & 8022.9554 & 0.924843 & 1.15e-04 & 0.230756 & 0 & True & 0.200 & 0.200 \\
5 & 13 & 8025.2952 & 8026.0461 & 0.750875 & 9.36e-05 & 0.187807 & 1 & False & 0.200 & 0.200 \\
5 & 14 & 7813.9142 & 7814.7794 & 0.865181 & 1.11e-04 & 0.218520 & 1 & False & 0.200 & 0.200 \\
5 & 15 & 8153.2501 & 8154.2186 & 0.968536 & 1.19e-04 & 0.239327 & 1 & False & 0.200 & 0.200 \\
5 & 16 & 7990.2803 & 7991.0772 & 0.796918 & 9.97e-05 & 0.198906 & 0 & True & 0.800 & 0.800 \\
5 & 17 & 8251.5202 & 8252.4800 & 0.959828 & 1.16e-04 & 0.236166 & 0 & True & 0.000 & 0.000 \\
5 & 18 & 8060.0538 & 8060.2796 & 0.225883 & 2.80e-05 & 0.088657 & 0 & True & 0.800 & 0.800 \\
5 & 19 & 8134.4239 & 8135.3982 & 0.974246 & 1.20e-04 & 0.236712 & 1 & False & 0.000 & 0.000 \\
\midrule
6 & 0 & 8008.2166 & 8009.9142 & 1.697554 & 2.12e-04 & 0.287239 & 1 & False & 0.167 & 0.167 \\
6 & 1 & 7954.0204 & 7955.7483 & 1.727967 & 2.17e-04 & 0.291775 & 2 & False & 0.167 & 0.167 \\
6 & 2 & 8515.2238 & 8516.9634 & 1.739556 & 2.04e-04 & 0.290807 & 1 & False & 0.333 & 0.333 \\
6 & 3 & 8041.5249 & 8043.2384 & 1.713437 & 2.13e-04 & 0.289196 & 2 & False & 0.167 & 0.167 \\
6 & 4 & 7954.5492 & 7956.2899 & 1.740700 & 2.19e-04 & 0.292520 & 1 & False & 0.167 & 0.167 \\
6 & 5 & 8081.4287 & 8083.1091 & 1.680414 & 2.08e-04 & 0.284768 & 1 & False & 0.167 & 0.167 \\
6 & 6 & 7919.1246 & 7920.7601 & 1.635555 & 2.06e-04 & 0.278385 & 1 & False & 0.167 & 0.167 \\
6 & 7 & 8341.3444 & 8343.0768 & 1.732363 & 2.08e-04 & 0.290533 & 1 & False & 0.167 & 0.167 \\
6 & 8 & 7780.0505 & 7783.9000 & 3.849454 & 4.94e-04 & 0.513606 & 2 & False & 0.667 & 0.667 \\
6 & 9 & 7758.1722 & 7758.4619 & 0.289658 & 3.73e-05 & 0.106511 & 0 & True & 0.333 & 0.333 \\
6 & 10 & 8025.2673 & 8026.9806 & 1.713292 & 2.13e-04 & 0.288591 & 1 & False & 0.000 & 0.000 \\
6 & 11 & 7937.6258 & 7939.3601 & 1.734233 & 2.18e-04 & 0.290941 & 1 & False & 0.167 & 0.167 \\
6 & 12 & 8029.2818 & 8031.0097 & 1.727883 & 2.15e-04 & 0.290740 & 1 & False & 0.167 & 0.167 \\
6 & 13 & 8032.5516 & 8034.2505 & 1.698886 & 2.11e-04 & 0.286380 & 1 & False & 0.167 & 0.167 \\
6 & 14 & 7821.8490 & 7823.5042 & 1.655203 & 2.12e-04 & 0.280963 & 2 & False & 0.167 & 0.167 \\
6 & 15 & 8161.4286 & 8163.1592 & 1.730598 & 2.12e-04 & 0.290538 & 1 & False & 0.167 & 0.167 \\
6 & 16 & 7996.1751 & 7997.7766 & 1.601449 & 2.00e-04 & 0.272186 & 0 & True & 0.667 & 0.500 \\
6 & 17 & 8260.1017 & 8261.8398 & 1.738092 & 2.10e-04 & 0.290642 & 0 & True & 0.000 & 0.000 \\
6 & 18 & 8067.8967 & 8068.3242 & 0.427497 & 5.30e-05 & 0.126487 & 0 & True & 0.667 & 0.667 \\
6 & 19 & 8142.1331 & 8143.8672 & 1.734151 & 2.13e-04 & 0.290278 & 2 & False & 0.000 & 0.000 \\

\end{longtable}
\newpage
\section{Appendix B: Raw Orbit Invariant Benchmark Data}
\begin{longtable}{r r S[table-format=1.1] c S[table-format=1.6] S[table-format=1.4] S[table-format=1.4]}
\caption{Complete Per-Trial Identifiability Benchmark Data ($T=2512$ days).} \label{tab:identifiability_raw} \\
\toprule
{$k$} & {Trial} & {$d_{\text{orbit}}$} & {MAP Agree} & {$\text{TV}(q^{\text{sec}}, q^{\text{naive}})$} & {Sector Acc$_{\text{orbit}}$} & {Naive Acc$_{\text{orbit}}$} \\
\midrule
\endfirsthead

\multicolumn{7}{c}{\tablename\ \thetable\ -- \textit{Continued from previous page}} \\
\toprule
{$k$} & {Trial} & {$d_{\text{orbit}}$} & {MAP Agree} & {$\text{TV}(q^{\text{sec}}, q^{\text{naive}})$} & {Sector Acc$_{\text{orbit}}$} & {Naive Acc$_{\text{orbit}}$} \\
\midrule
\endhead

\midrule
\multicolumn{7}{r}{\textit{Continued on next page}} \\
\endfoot

\bottomrule
\endlastfoot

2 & 42 & 0.0 & True & 0.001106 & 0.5000 & 0.5000 \\
2 & 43 & 0.0 & True & 0.002717 & 0.5000 & 0.5000 \\
2 & 44 & 0.0 & True & 0.004971 & 0.5000 & 0.5000 \\
2 & 45 & 0.0 & True & 0.000668 & 0.5000 & 0.5000 \\
2 & 46 & 1.0 & False & 0.006275 & 0.5000 & 0.5000 \\
2 & 47 & 0.0 & True & 0.001226 & 1.0000 & 1.0000 \\
2 & 48 & 0.0 & True & 0.004666 & 1.0000 & 1.0000 \\
2 & 49 & 1.0 & False & 0.002240 & 0.5000 & 1.0000 \\
2 & 50 & 1.0 & False & 0.002322 & 0.5000 & 0.0000 \\
2 & 51 & 0.0 & True & 0.001805 & 1.0000 & 1.0000 \\
2 & 52 & 0.0 & True & 0.002636 & 1.0000 & 1.0000 \\
2 & 53 & 0.0 & True & 0.001254 & 1.0000 & 1.0000 \\
2 & 54 & 0.0 & True & 0.001264 & 0.5000 & 0.5000 \\
2 & 55 & 0.0 & True & 0.001105 & 0.5000 & 0.5000 \\
2 & 56 & 0.0 & True & 0.003466 & 0.5000 & 0.5000 \\
2 & 57 & 0.0 & True & 0.003407 & 0.5000 & 0.5000 \\
2 & 58 & 0.0 & True & 0.002382 & 1.0000 & 1.0000 \\
2 & 59 & 0.0 & True & 0.001976 & 1.0000 & 1.0000 \\
2 & 60 & 1.0 & False & 0.001709 & 0.5000 & 1.0000 \\
2 & 61 & 0.0 & True & 0.001540 & 0.0000 & 0.0000 \\
\midrule
3 & 42 & 0.0 & True & 0.012793 & 1.0000 & 1.0000 \\
3 & 43 & 0.0 & True & 0.055119 & 0.6667 & 0.6667 \\
3 & 44 & 0.0 & True & 0.113357 & 0.6667 & 0.6667 \\
3 & 45 & 1.0 & False & 0.054890 & 0.6667 & 1.0000 \\
3 & 46 & 0.0 & True & 0.004024 & 0.3333 & 0.3333 \\
3 & 47 & 1.0 & False & 0.037047 & 0.6667 & 0.3333 \\
3 & 48 & 0.0 & True & 0.082350 & 0.3333 & 0.3333 \\
3 & 49 & 1.0 & False & 0.104094 & 0.6667 & 0.3333 \\
3 & 50 & 1.0 & False & 0.070053 & 0.6667 & 0.6667 \\
3 & 51 & 0.0 & True & 0.160554 & 0.6667 & 0.6667 \\
3 & 52 & 1.0 & False & 0.079185 & 1.0000 & 0.6667 \\
3 & 53 & 0.0 & True & 0.090150 & 0.3333 & 0.3333 \\
3 & 54 & 1.0 & False & 0.006135 & 0.3333 & 0.3333 \\
3 & 55 & 1.0 & False & 0.163592 & 0.6667 & 0.6667 \\
3 & 56 & 1.0 & False & 0.167792 & 0.6667 & 0.3333 \\
3 & 57 & 0.0 & True & 0.062208 & 0.3333 & 0.3333 \\
3 & 58 & 1.0 & False & 0.010464 & 1.0000 & 0.6667 \\
3 & 59 & 1.0 & False & 0.031673 & 0.6667 & 0.3333 \\
3 & 60 & 1.0 & False & 0.108394 & 0.6667 & 0.3333 \\
3 & 61 & 1.0 & False & 0.075671 & 1.0000 & 0.6667 \\
\midrule
4 & 42 & 2.0 & False & 0.141726 & 0.7500 & 0.7500 \\
4 & 43 & 1.0 & False & 0.167447 & 0.7500 & 0.5000 \\
4 & 44 & 1.0 & False & 0.099549 & 0.7500 & 0.5000 \\
4 & 45 & 1.0 & False & 0.074008 & 1.0000 & 0.7500 \\
4 & 46 & 1.0 & False & 0.088050 & 0.7500 & 0.7500 \\
4 & 47 & 2.0 & False & 0.109098 & 0.7500 & 0.5000 \\
4 & 48 & 2.0 & False & 0.518354 & 0.7500 & 0.7500 \\
4 & 49 & 2.0 & False & 0.148137 & 0.5000 & 0.5000 \\
4 & 50 & 1.0 & False & 0.232516 & 0.7500 & 0.5000 \\
4 & 51 & 2.0 & False & 0.187232 & 1.0000 & 0.5000 \\
4 & 52 & 2.0 & False & 0.138037 & 0.7500 & 0.5000 \\
4 & 53 & 1.0 & False & 0.054878 & 0.5000 & 0.7500 \\
4 & 54 & 2.0 & False & 0.162020 & 0.7500 & 0.2500 \\
4 & 55 & 1.0 & False & 0.190450 & 0.7500 & 0.7500 \\
4 & 56 & 1.0 & False & 0.173169 & 0.5000 & 0.2500 \\
4 & 57 & 1.0 & False & 0.143505 & 0.2500 & 0.2500 \\
4 & 58 & 2.0 & False & 0.163400 & 0.2500 & 0.2500 \\
4 & 59 & 2.0 & False & 0.160045 & 0.7500 & 0.5000 \\
4 & 60 & 1.0 & False & 0.169686 & 0.7500 & 1.0000 \\
4 & 61 & 2.0 & False & 0.155823 & 0.5000 & 0.5000 \\
\midrule
5 & 42 & 1.0 & False & 0.180192 & 0.8000 & 0.6000 \\
5 & 43 & 2.0 & False & 0.247252 & 0.8000 & 0.4000 \\
5 & 44 & 2.0 & False & 0.183006 & 0.6000 & 0.4000 \\
5 & 45 & 1.0 & False & 0.218315 & 0.4000 & 0.2000 \\
5 & 46 & 1.0 & False & 0.202874 & 0.4000 & 0.4000 \\
5 & 47 & 2.0 & False & 0.211652 & 0.6000 & 0.2000 \\
5 & 48 & 1.0 & False & 0.088575 & 0.6000 & 0.8000 \\
5 & 49 & 1.0 & False & 0.200917 & 0.6000 & 0.4000 \\
5 & 50 & 1.0 & False & 0.127233 & 0.6000 & 0.6000 \\
5 & 51 & 2.0 & False & 0.205053 & 0.8000 & 0.6000 \\
5 & 52 & 2.0 & False & 0.198185 & 0.8000 & 0.4000 \\
5 & 53 & 1.0 & False & 0.073736 & 0.8000 & 1.0000 \\
5 & 54 & 2.0 & False & 0.060503 & 0.8000 & 0.6000 \\
5 & 55 & 2.0 & False & 0.114969 & 0.8000 & 0.6000 \\
5 & 56 & 1.0 & False & 0.160835 & 0.8000 & 0.8000 \\
5 & 57 & 1.0 & False & 0.178490 & 0.8000 & 0.8000 \\
5 & 58 & 1.0 & False & 0.102568 & 0.6000 & 0.8000 \\
5 & 59 & 2.0 & False & 0.221395 & 0.8000 & 0.6000 \\
5 & 60 & 2.0 & False & 0.230520 & 0.8000 & 0.4000 \\
5 & 61 & 2.0 & False & 0.168755 & 0.6000 & 0.6000 \\
\midrule
6 & 42 & 2.0 & False & 0.230282 & 0.3333 & 0.5000 \\
6 & 43 & 2.0 & False & 0.311139 & 0.6667 & 0.5000 \\
6 & 44 & 3.0 & False & 0.610028 & 0.6667 & 0.5000 \\
6 & 45 & 3.0 & False & 0.275672 & 0.6667 & 0.1667 \\
6 & 46 & 2.0 & False & 0.303203 & 0.5000 & 0.3333 \\
6 & 47 & 2.0 & False & 0.274975 & 0.6667 & 0.3333 \\
6 & 48 & 3.0 & False & 0.239771 & 0.6667 & 0.1667 \\
6 & 49 & 3.0 & False & 0.175228 & 0.5000 & 0.5000 \\
6 & 50 & 3.0 & False & 0.228720 & 0.6667 & 0.3333 \\
6 & 51 & 3.0 & False & 0.226944 & 0.6667 & 0.1667 \\
6 & 52 & 1.0 & False & 0.153043 & 0.8333 & 0.8333 \\
6 & 53 & 3.0 & False & 0.381418 & 1.0000 & 0.5000 \\
6 & 54 & 3.0 & False & 0.177104 & 0.6667 & 0.5000 \\
6 & 55 & 3.0 & False & 0.249660 & 0.6667 & 0.3333 \\
6 & 56 & 2.0 & False & 0.094545 & 0.6667 & 0.5000 \\
6 & 57 & 3.0 & False & 0.373449 & 1.0000 & 0.5000 \\
6 & 58 & 2.0 & False & 0.160996 & 0.8333 & 0.6667 \\
6 & 59 & 2.0 & False & 0.163688 & 0.8333 & 0.5000 \\
6 & 60 & 3.0 & False & 0.276190 & 0.8333 & 0.6667 \\
6 & 61 & 3.0 & False & 0.599117 & 0.5000 & 0.6667 \\
\end{longtable}
\end{document}